\documentclass[envcountsame]{llncs}
\usepackage{booktabs}
\usepackage{makecell}
\usepackage{amsmath}
\usepackage{amssymb}
\usepackage{amsfonts}
\usepackage{hyperref}
\usepackage{color}
\usepackage{xcolor}
\usepackage{mathbbol}
\usepackage{multicol}
\usepackage{relsize}
\usepackage{algorithm}
\usepackage{algorithmic}
\usepackage{enumitem}
\newlist{customitemize}{itemize}{3}
\setlist[customitemize,1]{label=RC\arabic{customitemizei}.}

\newcommand{\ignore}[1]{}

\spnewtheorem{ND}{New Direction}{\bfseries}{\itshape}
\spnewtheorem{result}{Result}{\bfseries}{\itshape}
\spnewtheorem{open}{Open Problems}{\bfseries}{\itshape}
\spnewtheorem{fact}{Fact}{\bfseries}{\itshape}
\spnewtheorem{cor}[theorem]{Corollary}{\bfseries}{\itshape}
\spnewtheorem{assm}[theorem]{Assumption}{\bfseries}{\itshape}
\spnewtheorem{prop}[theorem]{Proposition}{\bfseries}{\itshape}
\renewcommand{\Pr}{{\bf Pr}}
\newcommand{\E}{{\bf E}}
\newcommand{\Var}{{\bf Var}}
\newcommand{\cA}{{\cal A}}
\newcommand{\cD}{{\cal D}}

\newcommand{\RC}{{\rm{RC}}}
\newcommand{\Inf}{{\rm{Inf}}}
\newcommand{\Majority}{{\rm{Majority}}}

\begin{document}

\title{Classes Testable with $O(1/\epsilon)$ Queries  for $\epsilon$ Below an $n$-Independent Threshold}
\author{Nader H. Bshouty\thanks{Corresponding author.} \and George Haddad}
\authorrunning{N. H. Bshouty and G. Haddad}
\institute{Dept. of Computer Science\\
Technion, Haifa, Israel\\
\email{bshouty@ca.technion.ac.il}}

\maketitle
\pagestyle{empty}
\thispagestyle{empty}
\begin{abstract}
In this paper, we study classes of Boolean functions that are testable with $O(\psi+1/\epsilon)$ queries, where $\psi$ depends on the parameters of the class (e.g., the number of terms, the number of relevant variables, etc.), is independent of (or sublinear in) $1/\epsilon$, and is independent of the total number of variables $n$. In particular, when $\epsilon$ is sufficiently small so that $\epsilon=O(1/\psi)$, the query complexity is $O(1/\epsilon)$, matching the known tight lower bound $\Omega(1/\epsilon)$. In particular, if the structural parameters of the class are fixed and $\epsilon=1/\omega_n(1)$ (e.g., $1/\log n$), then the asymptotic query complexity is $O(1/\epsilon)$.

This result was previously known for classes of terms of size at most $k$ and exclusive OR functions of at most $k$ variables. In this paper, we extend this list to include the classes: $k$-junta, functions with Fourier degree at most $d$, $s$-sparse polynomials of degree at most $d$, and $s$-sparse polynomials. 

Additionally, we show that for any class $C$ of Boolean functions that depend on at most $k$ variables, if $C$ is properly exactly learnable, then it is testable with $O(1/\epsilon)$ queries for $\epsilon<1/\psi$, where $\psi$ depends on $k$ and independent of the total number of variables $n$.   
\end{abstract}

\section{Introduction}
Property testing of sets of objects was first introduced in the seminal works of Blum, Luby, and Rubinfeld~\cite{BlumLR93} and Rubinfeld and Sudan~\cite{RubinfeldS96}. Since then, it has evolved into a highly active area of research; see, for instance, the surveys and books~~\cite{BhattacharyyaY24,CzumajS06,GoldreichSurvey10,Goldreich17,Ron08,Ron09}. 

Let $C$ be a class of Boolean functions $f:\{0,1\}^n\to\{0,1\}$. A {\it property testing algorithm} for $C$ with $q$ queries is a randomized algorithm that, given a function $f$, can access $f$ via a black-box that returns $f(x)$ for any query $x\in\{0,1\}^n$. If $f\in C$, the algorithm, with probability at least~$2/3$, outputs ``Accept''. If $f$ is $\epsilon$-far from every function in $C$, i.e., for every $g\in C$, $\Pr_x[f(x)\not=g(x)]\ge \epsilon$ (under the uniform distribution), with probability at least $2/3$, the algorithm outputs ``Reject''.

The classes studied in this paper include subclasses of $k$-\textsc{Junta}, as well as classes of Boolean functions that can be approximated by $k$-juntas. These classes have been extensively studied in the property testing literature \cite{Blais09,BlaisK12,BlumL97,Bshouty20x,Bshouty22AO,Bshouty23,DiakonikolasLMSW11} \cite{ChakrabortyGM11,DiakonikolasLMORSW07,GoldreichGR98,GopalanOSSW11,ParnasRS02}.

Our goal is to identify classes of Boolean functions that are testable with $O(\psi+1/\epsilon)$ queries, where $\psi$ depends only on the structural parameters of the class, is independent of the number of variables $n$, and is either independent of or sublinear in $1/\epsilon$. In particular, when $\epsilon\le O(1/\psi)$, the query complexity becomes $O(1/\epsilon)$, matching the lower bound $\Omega(1/\epsilon)$~\cite{BshoutyG22,Eldar}.

For instance, in~\cite{Bshouty23}, Bshouty presented a property testing algorithm for the class of functions that can be expressed as an exclusive OR of at most $k$ variables. The query complexity of the algorithm is $O(k\log k+1/\epsilon)$. When $\epsilon\le O(1/(k\log k))$ and $k$ is independent of $n$ the query complexity of this algorithm is $O(1/\epsilon)$.

In this paper, we investigate whether this query complexity (i.e., $O(\psi+1/\epsilon)$) holds for other classes of Boolean functions. We demonstrate this for $k$-junta when $\epsilon<1/(k2^k)$, functions with Fourier degree at most $d$ when $\epsilon<1/\tilde \Theta(2^{2d})$, s-sparse polynomials of degree at most $d$ when $\epsilon<1/\tilde \Theta(2^ds)$, and $s$-sparse polynomials when $\epsilon<1/s^{8.422}$. 

Unless otherwise specified, all the algorithms presented in this paper run in linear time with respect to the number of variables $n$ and polynomial time in the number of queries. The table in Figure~\ref{Table1} summarizes these results alongside previously known results.

\renewcommand{\arraystretch}{2.0}
\begin{figure}[h!]
\begin{center}
\begin{tabular}{|l|c|c|c|}
\hline
{\bf Class of Functions} &
{\bf Query Complexity} &
{\bf $=O(1/\epsilon)$ for} &
{\bf Reference} \\
\hline\hline

\textsc{$k$-Junta}
&
$\displaystyle O\!\left(k\log k+\frac{k}{\epsilon}\right)\rule{0pt}{3.2ex}$
&
---
&
\cite{Blais09}
\\
\cline{2-4}

&
$\displaystyle O\!\left(k2^k+\frac{1}{\epsilon}\right)\rule{0pt}{3.2ex}$
&
$\displaystyle \epsilon\le\frac{1}{k2^k}\rule{0pt}{3.2ex}$
&
This Paper
\\
\hline

\makecell[l]{Functions with\\Fourier Degree $\le d$}
&
$\displaystyle
\tilde O\!\left(2^{2d}+\frac{2^d}{\epsilon}\right)
\rule{0pt}{3.2ex}$
&
---
&
\cite{Bshouty20x}
\\
\cline{2-4}

&
$\displaystyle
\tilde O(2^{2d})+O\!\left(\frac{1}{\epsilon}\right)
\rule{0pt}{3.2ex}$
&
$\displaystyle
\epsilon\le\frac{1}{\tilde\Theta(2^{2d})}
\rule{0pt}{3.2ex}$
&
This Paper
\\
\hline

\makecell[l]{$s$-Sparse Polynomial\\of Constant Degree}
&
$\displaystyle
\tilde O\!\left(\frac{s}{\epsilon}\right)
\rule{0pt}{3.2ex}$
&
---
&
\cite{Bshouty20x}
\\
\cline{2-4}

&
$\displaystyle
\tilde O(s)+O\!\left(\frac{1}{\epsilon}\right)
\rule{0pt}{3.2ex}$
&
$\displaystyle
\epsilon\le\frac{1}{\tilde\Theta(s)}
\rule{0pt}{3.2ex}$
&
This Paper
\\
\hline

\makecell[l]{$s$-Sparse Polynomial\\$\epsilon=1/s^\beta$}
&
$\displaystyle
\left(\frac{s}{\epsilon}\right)^{\frac{\log\beta}{\beta}+O(1/\beta)}
+\tilde O\!\left(\frac{s}{\epsilon}\right)
\rule{0pt}{4.2ex}$
&
---
&
\cite{Bshouty22AO}
\\
\cline{2-4}

&
$\displaystyle
\left(\frac{\tilde O(s^2)}{\epsilon}\right)^{\frac{\log\beta}{\beta}+O(1/\beta)}
+\tilde O(s)+O\!\left(\frac{1}{\epsilon}\right)
\rule{0pt}{4.2ex}$
&
$\displaystyle
\epsilon<\frac{1}{s^{8.422}}
\rule{0pt}{4.2ex}$
&
This Paper
\\
\hline

\end{tabular}
\end{center}
\caption{A table of the results. Here $\tilde O(r)=O(r\cdot\operatorname{polylog}(r))$.}
\label{Table1}
\end{figure}

For the class of $s$-sparse polynomials, the first term is bounded by $1/\epsilon$ whenever $\epsilon<1/s^{8.422}$. Therefore, for $\epsilon<1/s^{8.422}$, the query complexity is $O(1/\epsilon)$. 

Additional results are presented below. 

\begin{theorem}\label{Th1}
    Let $C\subseteq k$-\textsc{Junta} be a class that is closed under variable permutations\footnote{For every $f\in C$ and any permutation $\phi:[n]\to [n]$, $f(x_{\phi(1)},\ldots,x_{\phi(n)})\in C$.} and zero-one projections\footnote{For every $f\in C$, $i\in [n]$ and $\xi\in\{0,1\}$, $f(x_1,\ldots,x_{i-1},\xi,x_{i+1},\ldots,x_n)\in C$.}.
    If $C$ is exactly properly learnable\footnote{A learning algorithm that returns a representation $h\in C$ of a function equivalent to the target.} with $Q(n)$ queries, then there is a property testing algorithm for $C$ with
    $$q:=Q(O(k^2))+O\left(\frac{k\log^2 k}{\log\log k}\right)+O\left(\frac{1}{\epsilon}\right)$$
    queries. In particular, the query complexity is $q=O(1/\epsilon)$ for $$\epsilon\le \frac{1}{Q(O(k^2))+\tilde \Theta (k)}.$$
    If $C$ is exactly learnable\footnote{A learning algorithm that returns a representation of a Boolean function $h$ equivalent to the target $f$.}, the above result holds with an algorithm that runs in exponential time in~$k$. 
\end{theorem}
Although our result is restricted to classes that are closed under variable permutations and zero-one projections, this is not a real constraint, as all classes of functions considered in the literature for testing satisfy these properties.

Define $$\mu(C):=\min_{f\in C}\min_{i\in \RC(f)}\Pr_x[f_{|x_i\gets 0}(x)\not=f_{|x_i\gets 1}(x)]$$
where $\RC(f):=\{i\in[n]\mid x_i\mbox{ is a relevant variable in }f\}
$ denotes the set of {\it relevant coordinates} in~$f$, and $f_{|x_i\gets \xi}$ represents the function $f$ after substituting $\xi$ for $x_i$.

We prove
\begin{theorem}\label{Th2}
Let $C\subseteq k$-\textsc{Junta} be a class that is closed under variable permutations and zero-one projections.
    If $C$ is exactly  properly learnable with $Q(n)$ queries, then there is a property testing algorithm for $C$ with
    $$q:=Q(k)+O\left(\frac{k}{\mu(C)}+\frac{k\log^2 k}{\log\log k}\right)+O\left(\frac{1}{\epsilon}\right)$$
    queries. We have $q=O(1/\epsilon)$ for $$\epsilon\le \frac{1}{Q(k)+\tilde \Theta (k/\mu(C))}.$$

    If $C$ is exactly learnable, the above result holds with an algorithm that runs in exponential time in $k$.

\end{theorem}

These results can be applied to numerous classes of Boolean functions studied in the literature. Examples include $s$-Term $r$-DNF (DNF formulas with at most $s$ terms, each of size at most $r$), decision trees of size at most $s$, and branching programs of size at most $s$. 

Moreover, combined with the variable reduction developed in Section~\ref{SectionVR}, our results also extend to $s$-Term DNF\footnote{Unlike in learning, testing algorithms and lower bounds are generally not monotone with respect to inclusion of classes. That is, a tester for a class $C$ does not necessarily yield a tester for a subclass or superclass of $C$, and a lower bound for testing $C$ does not necessarily extend to its subclasses or superclasses.}, $s$-Term Monotone DNF, Boolean formulas of size at most $s$, and Boolean circuits of size at most $s$. For a broader overview of such classes, see the survey in~\cite{Bshouty18}.

\subsection*{Motivation}
The significance of our results can be summarized as follows. 

First, for several well-studied classes, including functions of Fourier degree at most $d$, $s$-sparse polynomials of constant degree, and general $s$-sparse polynomials, our algorithms substantially improve the best previously known query complexities for all values of $\epsilon$.


Second, our results also imply that, for the above classes (and for the classes $s$-Term DNF, $s$-Term Monotone DNF, Boolean formulas of size at most $s$, and Boolean circuits of size at most $s$) no lower bound of the form $\Omega(\phi/\epsilon)$ can hold for any unbounded function $\phi$ depending only on~$k,s,d$ or $r$.

Third, our results establish a new connection between exact learning and property testing by showing that every exactly learnable subclass of $k$-\textsc{Junta} admits an optimal property tester in the small-$\epsilon$ regime.

Finally, our results identify a broad family of natural classes of Boolean functions for which, once $\epsilon$ is sufficiently small, the query complexity becomes independent of all structural parameters defining the class. Thus, $1/\epsilon$ is the only parameter governing the query complexity. In particular, our tester for $k$-\textsc{Junta} has query complexity $O(k2^k+1/\epsilon)$, and therefore, whenever $\epsilon=O(1/(k2^k))$, its query complexity becomes $O(1/\epsilon)$, matching the trivial lower bound. The same phenomenon holds for every exactly learnable subclass of $k$-\textsc{Junta}, yielding a broad family of classes that admit optimal testers in the small-$\epsilon$ regime, including $s$-Term $r$-DNF, decision trees of size at most $s$, decision trees of depth $d$, and branching programs of size at most $s$. More generally, if the structural parameters of the class are fixed and $\epsilon=1/\omega_n(1)$ (e.g., $\epsilon=1/\log n$), then the asymptotic query complexity is $O(1/\epsilon)$.

\subsection{Previous Technique}
Diakonikolas et al.~\cite{DiakonikolasLMORSW07} were the first to develop a general framework for property testing a broad family of Boolean function classes, including subclasses of $k$-\textsc{Junta} and classes that can be approximated by $k$-juntas. Their framework first randomly partitions the variables into $\mathrm{poly}(m/\epsilon)$ blocks such that, with high probability, each block contains at most one influential variable, that is, a variable whose value has a significant effect on the value of the function, where $m$ is the number of influential variables. For the classes under consideration, $m$ depends only on the parameters of the representation class and is independent of the number of variables $n$. Consequently, this partitioning is the key ingredient that eliminates the dependence of the query complexity on~$n$.

They then identify the influential blocks (i.e., blocks that contain influential variables) and, for a uniformly random assignment, recover the value of the influential variable in each influential block. This enables them to simulate random examples of a junta approximation of the target function, which are used by a PAC learning algorithm for the underlying representation class. Consequently, they obtain property testers with query complexity of the form $\phi/\epsilon^2$ and $\phi/\epsilon$, where $\phi$ depends on the parameters of the representation class, but whose running time is generally exponential in those parameters.

Bshouty~\cite{Bshouty20x} refined this framework by replacing the reduction to PAC learning with a reduction to membership-query learning. As in~\cite{DiakonikolasLMORSW07}, the variables are first partitioned into blocks so that, with high probability, each influential block contains exactly one influential variable. Again, the number of blocks is $\mathrm{poly}(m/\epsilon)$ and is therefore independent of the number of variables $n$. 

The key idea is to define a function $F$ obtained from the target function $f$ by replacing all the variables in each influential block with a single representative variable, using a different variable for each block. The resulting function $F$ depends on $\mathrm{poly}(m/\epsilon)$ variables and is therefore independent of the total number of variables~$n$. Bshouty showed that if $f$ belongs to the target class, then so does $F$, whereas if $f$ is $\epsilon$-far from the class, then either $F$ is $(\epsilon/2)$-far from $f$ or $F$ is $(\epsilon/2)$-far from the class. 
Moreover, every query to $F$ can be simulated using a query to $f$. This yields a reduction from property testing the original class to property testing the reduced class. Consequently, if the reduced class is learnable (or testable) in polynomial time, then so is the original class. 

The reduction also improves the query complexity compared to~\cite{DiakonikolasLMORSW07}; however, the resulting query complexity is still of the form $\phi/\epsilon$, where $\phi$ depends on the parameters of the class.

\subsection{Our Technique}
Our technique in this paper is similar to that in~\cite{Bshouty20x}, but it provides significantly simplified algorithms that achieve better results. Furthermore, this paper introduces an important advancement by identifying classes that can be tested with $O(\psi+1/\epsilon)$ queries, where $\psi$ is independent of the number of variables $n$. 

Let $C$ be a class of Boolean functions $f:\{0,1\}^n\to\{0,1\}$ that is closed under variable permutations
and zero-one projections.
We first reduce property testing of the class $C$ of Boolean functions over the variables $x_1,x_2,\ldots,x_n$ to property testing the class $C[m]$, consisting of all functions in $C$ that depend on a subset of the variables $x_1,x_2,\ldots,x_m$, where $m$ depends only on the parameters of the class and is independent of the number of variables $n$. This is similar to the reductions in~\cite{Bshouty19,DiakonikolasLMORSW07}. A key technical innovation lies in substituting fully independent random assignments with pairwise independent ones in simulating tests over the reduced function. This results in a significant improvement in query complexity while preserving the correctness guarantees of the testing procedure. This is detailed below.

\subsubsection{Classes of Functions that Depend on $k$ Variables}

Let $C$ be a class of Boolean functions that depend on at most $k$ variables, where $k$ is independent of $n$.
In Section~\ref{SectionRCV}, we present several algorithms that identify a set of influential variables. That is, ignoring the other variables in the function\footnote{By randomly uniformly substituting $0$ and $1$ with equal probability in each non-influential variable.} results in a Boolean function $f'$ that is $(\epsilon/2)$-close to the function $f$, and if $f\in C$, then $f'\in C$. Additionally, for each influential variable, the algorithm provides a witness demonstrating its relevance -- an assignment where altering the variable's value changes the function's output. This reduces the problem of testing any Boolean functions to testing functions with at most $k$ influential variables, along with witnesses for each.

Some of these algorithms are similar to those in~\cite{Bshouty20x}. However, the algorithms presented in this paper are significantly simpler and achieve better query complexity. Additionally, the other algorithms introduced here are new and rely on reducing learning algorithms to the problem of identifying influential variables and their corresponding witnesses.  

In Section~\ref{SectionBV}, we extend the approach from Section~\ref{SectionRCV} to handle ``blocks'' of variables. The function variables are randomly partitioned into blocks, where each block is a subset of the variables. The number of blocks is determined by the parameters of the class being tested, is independent of the total number of variables, and is chosen to be large enough to ensure that, if\footnote{$f'$ is the function resulting from the reduction described in the previous paragraph.} $f'\in C$, then, with high probability (probability close to $1$), each block contains at most one influential variable. Additionally, for every block containing an influential variable, the algorithm provides a witness: an assignment where flipping all the bits in the block results in a different output for the function $f'$. 

This procedure is identical to the one presented in~\cite{Bshouty20x}. We include this section in the paper for completeness. The key idea is to treat each block as a single variable and apply the algorithm from Section~\ref{SectionRCV} to identify the influential blocks and their corresponding witnesses.

In Section~\ref{SectionSC}, we use the witnesses of the influential blocks and the known technique of the self-corrector to construct the following procedure. This procedure, for any assignment $a=(a_1,\ldots,a_n)$ and any influential block containing an influential variable $x_j$ (which is not known to the algorithm), identifies $a_j$. This procedure is much simpler than the one presented in~\cite{Bshouty20x}, while maintaining the same query complexity.

To do so, we use the witness for the $j$th block to find (using $f'$) a function $g_j$ that can be queried and is close to $x_{j}$. Then, by applying the self-corrector technique,  $a_{j}$ can be identified with high probability using $O(1)$ queries.

Then, in Section~\ref{SectionTester}, we present the reduction from property testing $C$ to property testing $C[m]$, where $m$ is the number of blocks. Since the number of blocks is determined by the parameters of the class being tested and is independent of the total number of variables $n$, the query complexity of property testing $C[m]$ is also independent of $n$. The reduction is as follows:

Let ${\cal A}$ be the property testing algorithm for $C[m]$. The main idea is similar to the reduction in~\cite{Bshouty20x}, but for one of the tests, we use an algorithm with pairwise independent assignments instead of fully independent assignments, allowing the testing algorithm to achieve the same result with fewer queries, as will be explained next.

First, we use the algorithms in Section~\ref{SectionVR}, Section~\ref{SectionRCV}, and Section~\ref{SectionBV} to identify influential blocks and obtain a function $f'$ that is $(\epsilon/2)$-close to $f$. This reduces the problem to testing $f'$.
If $f\in C$, then, with high probability, each block contains at most one influential variable. A block that contains an influential variable is referred to as an {\it influential block}. 

Now, define a function $F$ that is equal to $f'$, where all the variables in each influential block are replaced with the value of the influential variable in that block. That is, if the influential blocks are $X_1,\ldots,X_m$ and the influential variables of the blocks are $x_{\tau_1},\ldots,x_{\tau_m}$, respectively, then for all $i\in [m]$, substitute $x_{\tau_i}$ for each $x_j\in X_i$ in $f'$ to obtain $F(x_{\tau_1},\ldots,x_{\tau_m})$. 

Note that the algorithm does not know the influential variables $x_{\tau_i}$, and therefore, querying of $F$ is not straightforward. However, it can query $\tilde F:=F(x_1,\ldots,x_m)$ with a single query to $f'$. This is because $F(x_1,\ldots,x_m)$ is the function obtained by substituting $x_i$ for each $x_j\in X_i$ in $f'$ for all $i\in [m]$. 

If $f'\in C$, it is clear that\footnote{This follows from the fact that each block contains at most one influential variable.} $F=f'$, and therefore $F\in C$ and\footnote{This follows from the fact that $C$ is closed under variable permutations.} $\tilde F\in C$. On the other hand, if $f'$ is $\epsilon/2$-far from every function in $C$, then either $f'$ is $\epsilon/4$-far from $F$, or $F$ is $\epsilon/4$-far from every function in $C$. Now, $F$ is $\epsilon/4$-far from every function in $C$  if and only if $\tilde F$ is $\epsilon/4$-far from every function in\footnote{This follows from the fact that $C$ is closed under variable permutations.} $C$. Therefore, it remains to perform two tests:
\begin{enumerate}
    \item Test if $F$ is $\epsilon/4$-far from $f'$.
    \item Test if $\tilde F$ is $\epsilon/4$-far from every function in $C$. 
\end{enumerate}

If one of the tests indicates that the condition is satisfied, the testing algorithm rejects; otherwise, it accepts. Recall that we assumed the existence of an algorithm ${\cal A}$ for testing $C[m]$, and that $\tilde F\in C[m]$ can be queried. Therefore, item~2 is accomplished by running ${\cal A}$ on~$\tilde F$. 

To perform the test in item~1, Bshouty's algorithm in~\cite{Bshouty20x} uses $O(1/\epsilon)$ fully independent random uniform assignments $a^{(i)}$ to check if $F(a^{(i)}_{\tau_1},\ldots,a^{(i)}_{\tau_m})=f'(a^{(i)})$. To achieve this, we can employ the algorithm from Section~\ref{SectionSC}, which extracts, for any assignment $a^{(i)}$, the values of the entries corresponding to the influential variables. This process requires $O(m/\epsilon)$ queries, as described in~\cite{Bshouty20x}. 

In this paper, we test whether $F(x_{\tau_1},\ldots,x_{\tau_m})=f'(x)$ using pairwise independent assignments instead of fully independent ones. We choose $t=\log(1/\epsilon)+O(1)$ i.i.d. uniform distributed assignments $a^{(i)}$ and test if $F(b)=f(b)$ for all the points $b$ in the linear span of these assignments. The key saving comes from the fact that if $u=v+w$, then $u_{\tau_i}=v_{\tau_i}+w_{\tau_i}$, making it sufficient to determine only $a^{(i)}_{\tau_1},\ldots,a^{(i)}_{\tau_m}$ for all $i\in [t]$. This reduces the number of queries to $\psi(m)+O(1/\epsilon)$.  

This reduction is the main reason our testing algorithms have query complexity of the form $\psi+O(1/\epsilon)$ rather than the more common $\phi'/\epsilon$. Combined with the previous improved reductions developed in this paper, which make $\psi$ independent of $\epsilon$ for almost all of the classes considered, this yields testers with query complexity $O(1/\epsilon)$ whenever $\epsilon$ is sufficiently small.

\subsubsection{Other Classes}
In Section~\ref{SectionVR}, we study classes containing Boolean functions of the form $f=g(T_1,\ldots,T_s)$, where $g$ is any Boolean function and $T_i$ are terms. Notice that $f$ may depend on all the variables $x_1,x_2,\ldots, x_n$. 
We show that, without incurring additional queries, property testing of such classes can be reduced to property testing the sub-class of functions in $C$ that depend on at most $k=\tilde O(s^2)\log(1/\epsilon)$ variables. After this reduction, we use the above tester.

This is achieved by showing that, if for every variable $x_i$ in $f$, with probability $p=O(1/(s\log(s/\epsilon$ $)))$, we map $x_i$ to $0$ or $1$ with equal probability, then with high probability, we obtain a function $f'$ that is $(\epsilon/2)$-close to $f$. Consequently, testing $f'$ is equivalent to testing $f$. Moreover, with high probability, the number of variables in $f'$ is at most $k=\tilde O(s^2)\log(1/\epsilon)$. 

\section{Definitions and Preliminary Results}

Let $C$ be a class of Boolean functions $f:\{0,1\}^n\to \{0,1\}$. We say that $C$ is {\it closed under variable permutations} if, for every permutation $\phi:[n]\to [n]$ and every $f\in C$ we have $f(x_{\phi(1)},\cdots,x_{\phi(n)})\in C$. For $i\in [n]$ and $\xi\in\{0,1\}$, we define $f_{|x_i\gets \xi}=f(x_1,x_2,\ldots,x_{i-1},\xi,x_{i+1},\ldots,x_n)$. We say that $C$ is {\it closed under zero-one projection} if, for every $f\in C$, $\xi\in\{0,1\}$ and $i\in [n]$, we have $f_{|x_i\gets \xi}\in C$.

A {\it term} is a conjunction of variables and negated variables. We define the class of $s$-\textsc{Term Function} to be the class of all functions of the form $f(T_1,\ldots,T_s)$, where $f:\{0,1\}^s\to\{0,1\}$ is any Boolean function, and $T_i$, $i\in[s]$, are terms over $\{x_1,x_2,\ldots,x_n\}$. This class is closed under variable permutations and zero-one projection. For a class of Boolean functions $C$, we denote by $C[t]$ the class of all the functions in $C$ that depend on a subset of the coordinates $[t]$.

We say that $x_i$ is a {\it relevant variable} in $f$ if $f$ depends on $x_i$, i.e., $f_{|x_i\gets 0}\not\equiv f$. Similarly, $i$ is called {\it relevant coordinate} in $f$ if $x_i$ is a relevant variable in $f$. The class $k$-\textsc{Junta} is the class of all Boolean functions $f:\{0,1\}^n\to \{0,1\}$ with at most $k$ relevant variables.  

Let $[n]=\{1,2,\ldots,n\}$. For $X\subset [n]$ we denote by $\{0,1\}^X$
the set of all binary strings of
length $|X|$, with coordinates indexed by $i\in X$. For $x\in \{0,1\}^n$ and $X\subseteq [n]$, we write $x_X\in\{0,1\}^{X}$ to denote the projection of $x$ onto the coordinates in $X$. We denote by $1^X$ and $0^X$ the all-one and all-zero strings in $\{0,1\}^{X}$, respectively. For a variable $x_i$, $x_i^X$ is the vector with coordinates in $X$ with all coordinates are equal to $x_i$.
When we write $x_I=0$ we mean $x_I=0^I$. 

For $X_1,X_2\subseteq [n]$ where $X_1\cap X_2=\emptyset$ and $x\in \{0,1\}^{X_1}, y\in \{0,1\}^{X_2}$, we write $x\circ y$ to denote their concatenation, i.e.,
the string in $\{0,1\}^{X_1\cup X_2}$ that agrees with $x$ over the coordinates in $X_1$ and agrees with $y$ over the coordinates in~$X_2$. 
Note that $x\circ y=y\circ x$. 

For example, let $X=\{1,3\}$, $Y=\{2,4\}$, $a=(a_1,a_2,a_3,a_4)$ and $b=(b_1,b_2,b_3,b_4)$. Then $a_X\circ b_Y=(a_1,b_2,a_3,b_4)$, $1^X\circ a_Y=(1,a_2,1,a_4)$, $x^Y\circ 0^X=(0,x,0,x)$, and $a_1^X\circ b_3^Y=(a_1,b_3,a_1,b_3)$.

Given $f,g:\{0,1\}^n\to \{0,1\}$, we say that $f$ is $\epsilon$-{\it close to $g$} if $\Pr_{x\sim U}[f(x)\not=g(x)]\le \epsilon$, where $x\sim U$ means $x$ is chosen from $\{0,1\}^n$ according to the uniform distribution. 
We say that $f$ is $\epsilon$-{\it far from $g$ } if $\Pr_{x\sim U}[f(x)\not=g(x)]\ge \epsilon$.
For a class of Boolean functions $C$, we say that $f$ is $\epsilon$-{\it far from every function in $C$} if, for every $g\in C$, $f$ is $\epsilon$-far from $g$. We will simply write $\Pr_{x}[\cdot]$ for $\Pr_{x\sim U}[\cdot]$.

For $S\subset [n]$, we define the {\it influence of the set $S$ on $f$} as 
$$\Inf_f(S)=2\Pr_{x,y}[f(x_{\bar S}\circ y_S)\not= f(x)].$$
The following result is from~\cite{FischerKRSS02}.
\begin{lemma}\label{InfLemma}
    For all $S,T\subset [n]$ and any Boolean function $f$,
    $$\Inf_f(S)\le \Inf_f(S\cup T)\le \Inf_f(S)+\Inf_f(T).$$
\end{lemma}

The following Lemma follows from Chernoff's bound.
\begin{lemma}\label{distinguish} Let $\alpha_2>\alpha_1$ and $\eta<1$ be non-negative constants.
    There is an algorithm {\bf Distinguish} $(\Pr_{x\sim \cD}[A(x)],$ $\alpha_1\epsilon,\alpha_2\epsilon)$ that draws 
    $m=O({1}/{\epsilon})$ samples according to the distribution $\cD$ and,
    with probability at least $1-\eta$, outputs $1$ if $\Pr_{x\sim \cD}[A(x)]>\alpha_2\epsilon$ and $0$ if $\Pr_{x\sim \cD}[A(x)]<\alpha_1\epsilon$.
\end{lemma}

We say that the Boolean function $f:\{0,1\}^n\to \{0,1\}$ is a literal if $f\in \{x_1,\ldots,x_n,\overline{x_1},\ldots,\overline{x_n}\}$, where $\overline{x}$ is the negation of $x$. The following is from~\cite{ParnasRS06}. 
\begin{lemma}\label{TestLiteral}
     There is a testing algorithm {\bf TestLiteral} for the class of literals $\{x_i,\overline{x_i}|i\in [n]\}$ with $O(1/\epsilon)$ queries. 
\end{lemma}

Throughout this paper, $\eta$ will represent a sufficiently small constant.

\section{Variable Reducibility}\label{SectionVR}
Recall the class of $s$-\textsc{Term Function}, which consists of all functions of the form $f(T_1,\ldots,T_s)$, where $f:\{0,1\}^s\to\{0,1\}$ is any Boolean function, and $T_i$, $i\in[s]$, are terms over $\{x_1,x_2,\ldots,x_n\}$. 

In Appendix~\ref{VariableReducibility}, we prove.
\begin{lemma}\label{sTF01}
    Let $C\subset s$-\textsc{Term Function} be a class that is closed under zero-one projection. If there is a testing algorithm for $C\cap(\tilde O(s^2))$-\textsc{Junta} with $q(\epsilon)$ queries, then there is a testing algorithm for $C$ with $q(\epsilon/2)+O(1/\epsilon)$ queries.
\end{lemma}

\section{Relevant Coordinates Verifiers}\label{SectionRCV}
In this section, we study algorithms that, given a function $f\in C$ accessible via a black box, return a small set of relevant variables such that ignoring the other variables results in a function close to $f$. The algorithms also provide evidence supporting the relevance of these variables.

An assignment $a$ is called a {\it witness} for the relevant coordinate $i$ in $f$ if it satisfies the condition $f(a) \neq f(b)$, where $b$ is the assignment that differs from $a$ only in coordinate $i$.

We now define the relevant coordinate verifier problem (RCV problem). 
A class $C$ is said to be {\it $k$-relevant coordinates verifiable} with $q(n,\epsilon)$ queries if there exists a randomized algorithm ({\it RC-verifier}) that, given a black-box access to $f\in C$, runs in polynomial time and, with probability at least\footnote{Recall that $\eta$ is any small constant.} $1-\eta$, returns a set of relevant coordinates $V$ in $f$ and, for each relevant coordinate $j \in V$, an assignment $w^{(j)}$ such that $V$ and all $w^{(j)}$ satisfy:
\begin{enumerate}[label={RC\arabic*}, labelindent=1cm, leftmargin=2cm]
    \item\hspace{-.09in}. $|V|\le k$. \label{Trc1}
    \item\hspace{-.09in}. $\Pr_{x,y}[f(x_V\circ y_{\overline{V}})\not=f(x)]\le \epsilon.$ \label{Trc2}
    \item\hspace{-.09in}. For every $j\in V$, $w^{(j)}$ is a witness for the relevant coordinate $j$ in $f(x)$. \label{Trc3}
\end{enumerate}
In Appendix~\ref{RelevantCoordinatesVerifiers} we give several relevant coordinate verifiers.

\section{Blocks Verifiers}\label{SectionBV}
In this section, we study algorithms that, for {\bf any} Boolean function $f$ accessible via a black-box, either reject or return disjoint ``blocks'' of coordinates $X_1,\ldots, X_k\subseteq [n]$ such that: (1) Each block $X_i$ either contains one relevant coordinate in $f$ or is ``close'' to containing one relevant coordinate in $f$. (2) Ignoring the other coordinates $[n]\backslash \bigcup_{i\in [k]} X_i$ in $f$ results in a function close to $f$. If the algorithm rejects, then with high probability, $f$ is not in $C$. 

We say that a class $C$ is {\it $(k,\alpha)$-relevant blocks verifiable} with $q(n,\epsilon,\alpha)$ queries if there exists a randomized algorithm ({\it RB-verifier}) that, given a black-box to {\bf any} Boolean function $f:\{0,1\}^n\to\{0,1\}$, runs in polynomial time and, with probability at least $1-\eta$, either returns ``Reject''-- in which case $f\not\in C$ -- or returns
$k'\le k$ disjoint sets (blocks) $X_1,\ldots,X_{k'}\subseteq [n]$, assignments $a^{(j)}$, $j\in [k']$, and $u\in \{0,1\}^n$ that satisfy: For $X=\cup_{i=1}^{k'}X_i$,
\begin{enumerate}[label={RB\arabic*}, labelindent=1cm, leftmargin=2cm]
    \item\hspace{-.09in}. $\Pr_x[f(x_X\circ u_{\overline{X}})\not=f(x)]\le {\epsilon}.$\label{Trb1}
    \item\hspace{-.09in}. For every $j\in[k']$, there exists $\tau_j\in X_j$ such that $f(a^{(j)}_{[n]\backslash X_j}\circ x_{X_j})$ is $\alpha$-close to $\{x_{\tau_j},\overline{x_{\tau_j}}\}$.\label{Trb2}
    \item\hspace{-.09in}. If $f\in C$, then for every $j\in[k']$, $X_j$ contains exactly one relevant variable in $f(x)$, and $f(a^{(j)}_{[n]\backslash X_j}\circ x_{X_j})\in\{x_{\tau_j},\overline{x_{\tau_j}}\}$.\label{Trb3}
\end{enumerate} 

In Appendix~\ref{BlocksVerifiers} we give several block verifiers.

\section{Self Corrector}\label{SectionSC}
In Appendix~\ref{SelfCorrector}, we introduce a simple self-corrector, following the standard self-correction technique of Blum, Luby, and Rubinfeld~\cite{BlumLR93}. The self-corrector has black-box access to a function $G(x)$ that is $\alpha$-close to $\{x_i,\overline{x_i}\}$. Given any assignment $a$, it returns $a_i$ with high probability.

\section{Testing Algorithm via Relevant Block Verifier} \label{SectionTester}
In this section, we construct testing algorithms for $C$ over $n$ variables from relevant block verifiers and testing algorithms for the class $C$ over a small number of variables. 

Recall that for a class $C$ over $n$ variables, $C[k]$ is the class of functions whose relevant coordinates are a subset of $[k]$.

We prove the following.
\begin{lemma}\label{Gresult}
Let $K\ge k$ be two integers. Let $C\subseteq$\textsc{$K$-Junta} be a class that is closed under variable permutations and zero-one projection. Suppose $C[k]$ is testable with $Q(k,\epsilon)$ queries. If $C$ is $(k,\alpha)$-relevant blocks verifiable with $q(K,k,\epsilon,\alpha)$ queries, then there is a testing algorithm for $C$ with $$q(K,k,\epsilon/4,\alpha)+Q(k,\epsilon/4)+O\left(k\frac{\log \frac{1}{\epsilon}(\log k+\log\log \frac{1}{\epsilon})}{\log\frac{1}{\alpha}}+\frac{1}{\epsilon}\right)$$ queries.
\end{lemma}
\begin{proof}
Let $\cA(\epsilon)$ be a testing algorithm for $C[k]$ with $Q(k,\epsilon)$ queries. Let RB-Verify be a $(k,\alpha)$-relevant blocks verifier for $C$ with $q(K,k,\epsilon,\alpha)$ queries. Consider the testing algorithm in Algorithm~\ref{alg:A}.

\begin{algorithm}
\caption{Testing Algorithm for $C$.}
\label{alg:A}
\begin{algorithmic}[1]
    \STATE\label{RunV} Run algorithm RB-Verify$(f,n,K,k,\alpha,\epsilon/4)$.\label{VerA}
    \STATE {\bf if} RB-Verify rejects {\bf then} Reject.
    \STATE Let $X_1,\ldots,X_{k'}\subseteq [n]$, assignments $a^{(j)}$, $j\in [k']$, and assignment $u$ be the elements generated by RC-Verify, satisfying items~\ref{Trb1}-\ref{Trb3}. Let $X=\cup_{i=1}^{k'}X_i$.
    \STATE Test with $\cA(\epsilon/4)$ if $F(y_1,\ldots,y_{k'}):=f(y_1^{X_1}\circ \cdots \circ y_{k'}^{X_{k'}}\circ u_{\overline{X}})$ is $(\epsilon/4)$-close to $C$. \label{Rej010} \STATE {\bf if} $\cA$ rejects {\bf then} Reject.\label{Rej01}
    \STATE Choose $t=\log (1/\epsilon)+\log(1/\eta)+1$ uniformly at random $z^{(1)},\ldots,z^{(t)}\in \{0,1\}^n$.
    \FOR{$i\in [t]$ and $j\in [k']$}
        \STATE $\xi^{(i)}_j \gets \text{SelfCorrect}(f(a^{(j)}_{[n]\backslash X_j}\circ x_{X_j}), z^{(i)}, \alpha,\eta /(tk'))$.\label{scalg}
    \ENDFOR
    \STATE Let $\xi^{(i)}=(\xi^{(i)}_1,\cdots,\xi^{(i)}_{k'})$ for all $i\in [t]$.
    \FOR{each $\lambda\in \{0,1\}^t\backslash \{0^t\}$}
        \STATE {\bf if} $f\left(\left(\sum_{i=1}^t\lambda_i z^{(i)}\right)_{X}\circ u_{\overline{X}}\right) \not= F\left(\sum_{i=1}^t\lambda_i\xi^{(i)}\right)$ {\bf then} Reject\label{eeeee}
    \ENDFOR
    \STATE 
    Accept.
\end{algorithmic}
\end{algorithm}

\noindent
{\bf Completeness}: Suppose $f\in C$. In step~\ref{RunV}, the algorithm RB-Verify$(f,n,K,k,$ $\alpha,\epsilon/4)$, with probability at least $1-\eta$, returns $k'\le k$ disjoint sets $X_1,\ldots,X_{k'}\subseteq [n]$, assignments $a^{(j)}$ for $j\in [k']$, and $u\in \{0,1\}^n$ that satisfy: For $X=\cup_{i=1}^{k'}X_i$,
\begin{enumerate}
    \item $\Pr_x[f(x_X\circ u_{\overline{X}})\not=f(x)]\le {\epsilon/4}.$\label{rbf1}
    \item For every $j\in[k']$, there exists $\tau_j\in X_j$ such that $f(a^{(j)}_{[n]\backslash X_j}\circ x_{X_j})\in\{x_{\tau_j},\overline{x_{\tau_j}}\}$ and $X_j$ contains one relevant variable in $f$.\label{rbf2}
\end{enumerate}
Since $f(x_X\circ u_{\overline{X}})\in C$, $C$ is closed under variable permutations, and each block $X_j$ contains exactly one relevant variable in $f(x)$, we also have $F(y_1,\ldots,y_{k'}):=f(y_1^{X_1}\circ \cdots \circ y_{k'}^{X_{k'}}\circ u_{\overline{X}})\in C$. 
Therefore, In step~\ref{Rej010}, ${\cal A}(\epsilon/4)$ accepts with probability at least $1-\eta$. 
Since $f(a^{(j)}_{[n]\backslash X_j}\circ x_{X_j})\in\{x_{\tau_j},\overline{x_{\tau_j}}\}$, by step~\ref{scalg} and Lemma~\ref{selfcorrector}, we have $\xi_j^{(i)}=z^{(i)}_{\tau_j}$ and $\xi^{(i)}=(z_{\tau_1}^{(i)},\ldots,z^{(i)}_{\tau_{k'}})$ for all $i\in [t]$ and $j\in [k']$. 
Since every $X_j$ contains one relevant variable in $f(x)$, and $f(a^{(j)}_{[n]\backslash X_j}\circ x_{X_j})\in\{x_{\tau_j},\overline{x_{\tau_j}}\}$, every $X_j$ contains {\it at most one} relevant variable in $f(x_X\circ u_{\overline{X}})$. If it contains one, it must be $x_{\tau_j}$. 
In particular, $$f(x_X\circ u_{\overline{X}})=f(x_{X_1}\circ \cdots \circ x_{X_{k'}}\circ u_{\overline{X}})=f(x_{\tau_1}^{X_1}\circ \cdots\circ x_{\tau_{k'}}^{X_{k'}}\circ u_{\overline{X}}).$$

Thus,
\begin{eqnarray*}
f\left(\left(\sum_{i=1}^t\lambda_i z^{(i)}\right)_{X}\circ u_{\overline{X}}\right) 
&=&f\left(\left(\sum_{i=1}^t \lambda_iz^{(i)}\right)_{X_1}\circ \cdots \circ \left(\sum_{i=1}^t \lambda_iz^{(i)}\right)_{X_{k'}}\circ u_{\overline{X}}\right)\\
&=&f\left(\left(\sum_{i=1}^t \lambda_iz^{(i)}_{\tau_1}\right)^{X_1}\circ \cdots \circ \left(\sum_{i=1}^t \lambda_iz^{(i)}_{\tau_{k'}}\right)^{X_{k'}}\circ u_{\overline{X}}\right)\\
 &=& f\left(\left(\sum_{i=1}^t \lambda_i\xi^{(i)}_{1}\right)^{X_1}\circ \cdots \circ \left(\sum_{i=1}^t \lambda_i\xi^{(i)}_{{k'}}\right)^{X_{k'}}\circ u_{\overline{X}}\right)\\
 &=& F\left(\left(\sum_{i=1}^t \lambda_i\xi^{(i)}_{1}\right), \cdots , \left(\sum_{i=1}^t \lambda_i\xi^{(i)}_{{k'}}\right)\right)\nonumber\\
&=&F\left(\sum_{i=1}^t\lambda_i\xi^{(i)}\right).\nonumber
\end{eqnarray*}
Hence, the algorithm does not reject in step~\ref{eeeee}. Therefore, with probability at least $1-2\eta$, the algorithm accepts.

\noindent
{\bf Soundness}: Suppose $f$ is $\epsilon$-far from every function in $C$. If RB-Verify$(f,n,K,k,$ $\alpha,\epsilon/4)$ in step~\ref{VerA} does not reject, then with probability at least $1-\eta$, it returns $k'\le k$ disjoint sets $X_1,\ldots,X_{k'}\subseteq [n]$, assignments $a^{(j)}$ for $j\in [k']$ and, $u\in \{0,1\}^n$ that satisfy: For $X=\cup_{i=1}^{k'}X_i$,
\begin{enumerate}
    \item $\Pr_x[f(x_X\circ u_{\overline{X}})\not=f(x)]\le {\epsilon/4}.$\label{rbf21}
    \item For every $j\in[k']$, there exists $\tau_j\in X_j$ such that $f(a^{(j)}_{[n]\backslash X_j}\circ x_{X_j})$ is $\alpha$-close to $\{x_{\tau_j},\overline{x_{\tau_j}}\}$.\label{rbf22}
\end{enumerate}
Therefore, with probability at least $1-\eta$, $f(x_X\circ u_{\overline{X}})$ is $(3\epsilon/4)$-far from every function in $C$. 
If $F(y_1,\ldots,y_{k'})$ is $(\epsilon/4)$-far from every function in $C$, then the algorithm, with probability at least $1-\eta$, rejects in steps~\ref{Rej010}-\ref{Rej01}. Therefore, with probability at least $1-\eta$, $F(y_1,\ldots,y_{k'})$ is $(\epsilon/4)$-close to $C$, and thus $F(x_{\tau_1},\ldots,x_{\tau_{k'}})$ is $\epsilon/4$-close to $C$.
Thus, with probability at least $1-2\eta$, $f(x_X\circ u_{\overline{X}})$ is $(\epsilon/2)$-far from $F(x_{\tau_1},\ldots,x_{\tau_{k'}})$. That is,
\begin{eqnarray}\label{Arg01}
    \Pr_x[f(x_X\circ u_{\overline{X}})\not=F(x_{\tau_1},\ldots,x_{\tau_{k'}})]\ge \frac{\epsilon}{2}.
\end{eqnarray}

Since $f(a^{(j)}_{[n]\backslash X_j}\circ x_{X_j})$ is $\alpha$-close to $\{x_{\tau_j},\overline{x_{\tau_j}}\}$, by Lemma~\ref{selfcorrector}, step~\ref{scalg}, and the union bound, with probability at least $1-\eta$, for all $i\in [t]$ and $j\in[k']$, we have $\xi_j^{(i)}=z^{(i)}_{\tau_j}$. 
Therefore, with probability at least $1-\eta$, we have
\begin{eqnarray}
    F\left(\sum_{i=1}^t \lambda_i \xi^{(i)}\right)&=&F\left(\sum_{i=1}^t \lambda_i\xi^{(i)}_{1},\cdots,\sum_{i=1}^t \lambda_i\xi^{(i)}_{{k'}}\right)\nonumber\\
    &=& F\left(\sum_{i=1}^t \lambda_i z_{\tau_1},\cdots,\sum_{i=1}^t \lambda_i z_{\tau_{k'}}\right).\label{Arg02}
\end{eqnarray}
By (\ref{Arg01}) and (\ref{Arg02}), for uniformly at random $z^{(1)},\ldots,z^{(t)}\in \{0,1\}^n$, with probability at least $1-3\eta$, for any $(\lambda_1,\ldots,\lambda_t)\in \{0,1\}^t\backslash \{0^t\}$, we have
$$\Pr\left[ f\left(\left(\sum_{i=1}^t\lambda_i z^{(i)}\right)_{X}\circ u_{\overline{X}}\right)\not= F\left(\sum_{i=1}^t \lambda_i \xi^{(i)}\right)\right]\ge \frac{\epsilon}{2}.$$
For $\lambda\in \{0,1\}^t\backslash \{0^t\}$ let $W_\lambda$ be the indicator variable of the event
$$f\left(\left(\sum_{i=1}^t\lambda_i z^{(i)}\right)_{X}\circ u_{\overline{X}}\right) \not= F\left(\sum_{i=1}^t\lambda_i\xi^{(i)}\right).$$
By Lemma~\ref{pifun} and since $t=\log (1/\epsilon)+\log(1/\eta)+1$, the probability that the algorithm does not reject in step~\ref{eeeee} is
$$\Pr\left[\sum_{\lambda\in\{0,1\}^t\backslash \{0^t\}}W_\lambda=0\right]\le \Pr\left[\left|\frac{\sum_{\lambda\in\{0,1\}^t\backslash \{0^t\}}W_\lambda}{2^t-1}-\frac{\epsilon}{2}\right|\le\frac{\epsilon}{2}\right]\le $$ $$\frac{1}{(2^t-1)\epsilon}\le \eta.$$
Therefore, with probability at least $1-4\eta$ the testing algorithm rejects.

\noindent
{\bf Query Complexity}: For the query complexity, RB-Verify$(f,n,K,k,\alpha,\epsilon/4)$ makes $q(K,k,\epsilon/4,\alpha)$ queries. The algorithm ${\cal A}(\epsilon/4)$ makes $Q(k,\epsilon/4)$ queries. By Lemma~\ref{selfcorrector}, SelfCorrect$(\cdot,\cdot,\alpha,\eta/(tk'))$ makes at most
$$O\left(tk'\frac{\log({tk'})}{\log\frac{1}{\alpha}}\right)=O\left(k\frac{\log \frac{1}{\epsilon}(\log k+\log\log \frac{1}{\epsilon})}{\log\frac{1}{\alpha}}\right)$$
queries. Step \ref{eeeee} makes $O(2^t)=O(1/\epsilon)$ queries.
\end{proof}

All the results are in Appendix~\ref{Results}.

\bibliography{TestingRef}

\newpage
\appendix
\section{Chernoff and Chebychev's Bound}
\begin{lemma}\label{Chernoff}{\bf Chernoff's Bound}. Let $X_1,\ldots, X_m$ be independent random variables taking values in $\{0, 1\}$. Let $X=\sum_{i=1}^mX_i$ denote their sum, and let $\mu = \E[X]$ denote the expected value of the sum. Then
\begin{equation}\label{Chernoff1}
\Pr[X\ge (1+\lambda)\mu]\le \left(\frac{e^{\lambda}}{(1+\lambda)^{(1+\lambda)}}\right)^{\mu}\le e^{-\frac{\lambda^2\mu}{2+\lambda}}\le 
\begin{cases}
e^{-\frac{\lambda^2\mu} {3}} & \mbox{if \ }  0< \lambda\le 1 
\\
e^{-\frac{\lambda \mu}{3}} & \mbox{if \ } \lambda>1.
\end{cases}
\end{equation}
In particular,
\begin{eqnarray}\label{Chernoff2}\Pr[X>\Lambda]\le \left(\frac{e\mu}{\Lambda}\right)^{\Lambda}.\end{eqnarray}

For $0\le \lambda\le 1$, we have
\begin{eqnarray}\label{Chernoff3}
\Pr[X\le (1-\lambda)\mu]\le \left(\frac{e^{-\lambda}}{(1-\lambda)^{(1-\lambda)}}\right)^{\mu}\le e^{-\frac{\lambda^2\mu}{2}}.
\end{eqnarray}
\end{lemma}

\begin{lemma} {\bf Chebyshev's Bound}. Let $X$ be a random variable with expectation $\mu$ and variance $\sigma^2=\Var[X]$. For any $k>0$, we have
$$\Pr[|X-\mu|\ge k]\le \frac{\Var[X]}{k^2}.$$
\end{lemma}

In the following lemma, we apply Chebyshev's bound to the sum of pairwise independent random variables.
\begin{lemma}\label{CIS}
    Let $X_1, X_2, \ldots, X_\ell$ be $\ell$ pairwise independent random variables that take values in $\{0,1\}$ with a common expectation $\mu'$. Then
    $$\Pr\left[\left|\frac{\sum_{i=1}^\ell X_i}{\ell}-{\mu'}\right|\ge \frac{\mu'}{c}\right]\le \frac{c^2}{\ell\mu'}.$$
\end{lemma}
\begin{proof}
    Let $X=X_1+X_2+\cdots+X_\ell$. Then $\mu:=\E[X]=\ell\mu'$ and $\Var[X_i]=\E[X_i^2]-\E[X_i]^2=\E[X_i]-\E[X_i]^2=\mu'(1-\mu')$. Since the variables are pairwise independent,    $$\Var[X]=\sum_{i=1}^\ell\Var[X_i]=\ell\mu'(1-\mu').$$
    By Chebychev's bound $$\Pr\left[\left|\frac{\sum_{i=1}^\ell X_i}{\ell}-{\mu'}\right|\ge \frac{\mu'}{c}\right]=\Pr\left[|X-\mu|\ge \frac{\mu}{c}\right]\le $$ $$\frac{c^2\ell\mu'(1-\mu')}{\mu^2}\le \frac{c^2\ell\mu'(1-\mu')}{\ell^2(\mu')^2}\le \frac{c^2}{\ell\mu'}.$$
\end{proof}

Recall that, $x\sim U$ indicates that $x$ is drawn uniformly at random from $\{0,1\}^n$.

In particular, we have.
\begin{lemma}\label{pifun}
Let $Y$ be a function from the Boolean vectors $\{0,1\}^n$ to the real numbers $\{0,1\}$. If $\Pr_{x\sim U}[Y(x)=1]= \epsilon$, then 
$$\Pr_{v^{(1)},\ldots,v^{(m)}\sim U}\left [\left|\frac{\sum_{\lambda\in\{0,1\}^m\backslash \{0^m\}} Y\left(\sum_{i=1}^m \lambda_i v^{(i)}\right)}{2^m-1}-\epsilon\right|\ge \frac{\epsilon}{c}\right]\le \frac{c^2}{(2^m-1)\epsilon}.$$
\end{lemma}
\begin{proof}
    The result follows from Lemma~\ref{CIS} and the observation that when $v^{(1)},\ldots,$ $v^{(m)}\sim U_n$, the random variables $\{Y\left(\sum_{i=1}^m \lambda_i v^{(i)}\right)\}_{\lambda\in\{0,1\}^m\backslash\{0^m\}}$, are pairwise independent.
\end{proof}

\section{Variable Reducibility}\label{VariableReducibility}
We say that a class of Boolean functions $C$ is $t$-{\it variable reducible} by the reduction $R$ with $q$ queries if $R$ is a polynomial-time randomized reduction that transforms any function $f$ into a new function $\hat f:=R(f,\epsilon)$ satisfying the following: For any small constant $\eta$,
\begin{enumerate}
\item\label{VR1} It is possible to simulate a black-box query to $\hat f$ with $q$ black-box queries to $f$. 

\hspace{-.43in}
If $f\in C$, then
    \item\label{VR2} $\hat f\in C$. 
    \item\label{VR3} With probability at least $1-\eta$, $\hat f$ depends on at most $t$ variables.
    \item\label{VR4} With probability at least $1-\eta$, $\Pr_x[\hat f(x)\not= f(x)]\le \epsilon$.
\end{enumerate}
We will simply write $\hat f=R(f)$ when $\epsilon$ is understood from the context.

We now prove.
\begin{lemma}\label{RedVar} Let $C$ be a class that is $t$-variable reducible by a reduction $R$ with $q'(\epsilon)$ queries.
If there is a testing algorithm for $C_t:=C\cap t-$\textsc{Junta} with $q(\epsilon)$ queries, then there is a testing algorithm for $C$ with $q'(\epsilon/4) (q(\epsilon/2)+O(1/\epsilon)))$ queries.
\end{lemma}
\begin{proof}
Let $TestC_t(\epsilon,\delta)$ be a testing algorithm for $C_t$ with $q(\epsilon)$ queries.
\begin{algorithm}
\caption{Testing algorithm for $C$}
\label{alg:property_testing}
\begin{algorithmic}[1]
    \STATE Let $\hat f=R(f,\epsilon/4)$.
    \STATE {\bf If} {{\bf Distinguish}$(\Pr[\hat f(x)\neq f(x)], \frac{\epsilon}{4}, \frac{\epsilon}{2}) = 1$}
   {\bf then} \label{Reject1A10}  Reject.  \label{abc20}  
    \STATE Run the testing algorithm $TestC_t(\epsilon/2,\eta)$ on $\hat f$ and return its output.
\end{algorithmic}
\end{algorithm}
 Consider the testing algorithm for $C$ in Algorithm~\ref{alg:property_testing}. In the algorithm, {\bf Distinguish} is the procedure that is defined in Lemma~\ref{distinguish}. It, with probability at least $1-\eta$, outputs $1$ if $\Pr_{x}[\hat f(x)\neq f(x)]>\epsilon/2$ and $0$ if $\Pr_{x}[\hat f(x)\neq f(x)]<\epsilon/4$.

Let $f$ be any Boolean function, and $\hat f=R(f,\epsilon/4)$. If $f\in C$, then by item~\ref{VR4}, with probability at least $1-\eta$, $\Pr_x[\hat f(x)\not=f(x)]\le \epsilon/4$. Therefore, by Lemma~\ref{distinguish}, the algorithm, with probability at least $1-2\eta$, does not reject in step~\ref{Reject1A10}.
Since, by item~\ref{VR2} and~\ref{VR3}, with probability at least $1-\eta$, $\hat f\in C_t$, with probability at least $1-2\eta$, $TestC_t(\epsilon/2,\eta)$ accepts. Therefore, if $f\in C$, algorithm\ref{alg:property_testing} accepts with probability at least $1-4\eta$. 

Now, let $f$ be a Boolean function that is $\epsilon$-far from every function in $C$. If $\Pr_x[\hat f(x)\not=f(x)]>\epsilon/2$, then by Lemma~\ref{distinguish}, with probability at least $1-\eta$ algorithm~\ref{alg:property_testing} rejects in step~\ref{abc20}. If $\Pr_x[\hat f(x)\not=f(x)]\le \epsilon/2$, then, since $f$ is $\epsilon$-far from every function in $C$, $\hat f$ is $(\epsilon/2)$-far from every function in $C$ and therefore $\hat f$ is $(\epsilon/2)$-far from every function in $C_t\subseteq C$. Thus,  $TestC_t(\epsilon/2,\eta)$ rejects with probability at least $1-\eta$. 

The query complexity follows from Lemma~\ref{distinguish} and the fact that every query to $\hat f$ requires $q'(\epsilon/4)$ queries to $f$.
\end{proof}
For $0 \leq p \leq 1$ and the variables $x=(x_1,\ldots,x_n)$, consider the random map $R_p(x)=y$, where for each $i\in [n]$ the value of $y_i$ is equal to $x_i$ with probability $1-p$, and equal to $0$ or $1$ with probability~$p/2$ each.

\begin{lemma}\label{PfT}
Let $p=\eta/(s\log(s/\epsilon))$. The class $C=$ $s$-\textsc{Term Function} is $\tilde O(s^2)$-{\it variable reducible} by the reduction $R_p$ with one query.
\end{lemma}
\begin{proof} To distinguish between the randomness of $x$ and $R_p$, denote $R_p^r$ as the reduction $R_p$ with the random seed $r$.

We show that for any function $F=f(T_1,T_2,\ldots,T_s)$, where $f:\{0,1\}^s\to \{0,1\}$ and $T_1,T_2,\ldots,T_s$ are terms, with probability at least $1-2\eta$, the function $\hat F=R^r_p(F)$ satisfies the following properties:
\begin{enumerate}
    \item\label{iii1} $\Pr_x[\hat F(x)\not=F(x)]\le \epsilon$.
    \item \label{iii2} $\hat F$ depends on at most $\tilde O(s^2)$ variables.
\end{enumerate}
The fact that $\hat F\in C$ and that a query to $\hat F$ can be simulated with one query to $F$ is straightforward.  

We now prove item \ref{iii1}. Fix a random seed $r$ (a fixed map). We have
\begin{eqnarray*}
\Pr_x[R^r_p(F)(x)\not=F(x)]&=&\Pr_x[f(R^r_p(T_1)(x),\ldots,R^r_p(T_s)(x))\not=f(T_1(x),\ldots,T_s(x))]\\
&\le & \sum_{i=1}^s \Pr_x[R^r_p(T_i)(x)\not=T_i(x)].
\end{eqnarray*}
Now, item~\ref{iii1} follows from the following:
\begin{claim} For any term $T$, with probability at least $1-\eta/s$, $\Pr_x[R^r_p(T)(x)\not=T(x)]\le \epsilon/s$.
\end{claim}
\begin{proof}
For a term $T$, denote by $|T|$ the size of $T$, i.e., the number of variables in $T$.

We first prove by induction that for any term $T$,
\begin{eqnarray}\label{PhiT}
    \phi(T):=\E_r[\Pr_{x}[R^r_p(T)(x)\not=T(x)]]\le \frac{|T|}{2^{|T|}}p.
\end{eqnarray}
The base case holds trivially. Suppose w.l.o.g. $T=x_1T'$, and $|T'|=|T|-1$. Then
\begin{eqnarray*}
    \E_r[\Pr_x[R^r_p(T)(x)=1]]&=& \frac{p}{2}\E_r[\Pr_x[R^r_p(T')(x)=1]]+\\ &&(1-p)\E_r[\Pr_x[x_1R^r_p(T')(x)=1]]\\
    &=&\frac{p}{2}\E_r[\Pr_x[R^r_p(T')(x)=1]]+\\&&\frac{1-p}{2}\E_r\left[\Pr_x[R^r_p(T')(x)=1]\right]\\
    &=&\frac{1}{2} \E_r[\Pr_x[R^r_p(T')(x)=1]].
\end{eqnarray*}
Therefore, for any term $T$, we have 
\begin{eqnarray}\label{1oTtt}
    \E_r[\Pr_x[R^r_p(T)(x)=1]]=\Pr_x[T(x)=1]=\frac{1}{2^{|T|}}.
\end{eqnarray} 
Now, by the definition of $R_p$ and (\ref{1oTtt}), we have
\begin{eqnarray*}
\E_r[\Pr_{x}[R^r_p(T)(x)\not=T(x)]]&=& (1-p)\E_r[\Pr_{x}[x_1R^r_p(T')(x)\not=x_1T'(x)]]+\\
&& \frac{p}{2} \E_r[\Pr_{x}[0\not=x_1T'(x)]]+\\&&
\frac{p}{2}\E_r[\Pr_{x}[R^r_p(T')(x)\not=x_1T'(x)]]\\
&=& \frac{1-p}{2}\phi(T')+\frac{p}{2^{|T|+1}}+\\ &&\frac{p}{4}\phi(T')+\frac{p}{4}\E_r[\Pr_x[R^r_p(T')(x)=1]]\\
&=& \frac{1-p}{2}\phi(T')+\frac{p}{2^{|T|+1}}+\frac{p}{4}\phi(T')+\frac{p}{4}\frac{1}{2^{|T|-1}}\\
&=& \left(\frac{1}{2}-\frac{p}{4}\right)\phi(T')+\frac{p}{2^{|T|}}\le\frac{\phi(T')}{2}+\frac{p}{2^{|T|}}\\
&\le& \frac{|T'|p}{2^{|T'|+1}}+\frac{p}{2^{|T|}} \mbox{\ \ \ By the induction hypothesis.}\\
&=& \frac{|T|}{2^{|T|}}p.
\end{eqnarray*}
This proves (\ref{PhiT}). 

We now distinguish between two cases.

\noindent 
{\bf Case I.} $|T|\le \log (s/\epsilon)$. The probability that $R^r_p(T)\equiv T$, i.e., the term does not change, is
$$(1-p)^{|T|}\ge 1-|T|p\ge 1-\frac{\eta}{s}.$$
Therefore, with probability at least $1-\eta/s$, we have $\Pr_x[R^r_p(T)(x)\not=T(x)]=0\le \epsilon/s$. 

\noindent
{\bf Case II.} $|T|> \log (s/\epsilon)$. 
Since, for $x>2$, $x/2^x$ is a monotone decreasing function, for  $|T|>\log(s/\epsilon)$, we have
$$\E_r[\Pr_{x}[R^r_p(T)(x)\not=T(x)]]\le \frac{|T|}{2^{|T|}}p\le \frac{\log(s/\epsilon)}{s/\epsilon}\frac{\eta}{s\log(s/\epsilon)}=\frac{\epsilon\eta}{s^2}.$$ 
By Markov's bound, with probability at least $1-\eta/s$, we have $\Pr_{x}[R^r_p(T)(x)\not=T(x)]\le \epsilon/s$. 

This completes the proof of the claim and, therefore, item~\ref{iii1}.
\end{proof}

Now, we prove item \ref{iii2}. The probability that a term of size greater than $k=(2s/\eta)\log(s/\epsilon)\ln(s/\eta)$ will not vanish under $R_p$ is at most
$$\left(1-\frac{p}{2}\right)^k\le e^{-pk/2}\le \frac{\eta}{s}.$$ Therefore, with probability at least $1-\eta$, all the terms in $F$ of size at least $k$ vanish. In particular, with probability at least $1-\eta$, the number of relevant variables in $F$ is at most $sk=\tilde O(s^2)$. 
\end{proof}

Now, Lemma~\ref{sTF01} follows from Lemma~\ref{RedVar} and~\ref{PfT}. 

\section{Relevant Coordinates Verifiers}\label{RelevantCoordinatesVerifiers}
We say that $C$ is {\it exactly learnable} with $q$ queries if there exists a randomized algorithm that for any $f\in C$, given black-box access to $f$, runs in polynomial time, makes $q$ queries, and, with probability at least $1-\eta$, outputs a hypothesis $h$ equivalent to $f$. If the output hypothesis $h$ belongs to $C$, then we say that $C$ is {\it properly exactly learnable} with $q$ queries. 

Our first result establishes a reduction from the RCV problem to the exact learning problem.
\begin{lemma}\label{ELtoRCV}
    Let $C\subseteq k$-\textsc{Junta} be a class of functions that is closed under zero-one projections. If $C$ is exactly learnable with $Q(n)$ queries, then $C$ is $k$-relevant coordinates verifiable with $Q(n)$ queries.
\end{lemma}
\begin{proof}
    Let $A(n,k)$ be an exact learning algorithm for $C$.
    We run $A(n,k)$ and, with probability at least $1-\eta$, obtain a hypothesis $h$ that is equivalent to the target $f$. 
    
    Now, for each variable $x_j$, we run $A(n,k)$ $O(\log n)$ times with the targets $h_{|x_j\gets 0}$ and $h_{|x_j\gets 1}$. 
    If $h_{|x_j\gets 0}(a)=h_{|x_j\gets 1}(a)$ for every query $a$ made by the algorithm, then with probability at least $1-\eta/n$, $f$ does not depend on $x_j$. Otherwise, we find an assignment $a$ such that $h_{|x_j\gets 0}(a)\not=h_{|x_j\gets 1}(a)$ and set $w^{(j)}\gets a$ as a witness for $x_j$.
\end{proof}
We say that $C$ is {\it learnable from} $H$ with $q$ queries if there exists a randomized algorithm that for any $f\in C$ and $\epsilon>0$, given black-box access to $f$, runs in polynomial time, makes $q$ queries, and, with probability at least $1-\delta$, outputs a hypothesis $h\in H$ that is $\epsilon$-close to $f$, i.e., $\Pr_x[f(x)\not=h(x)]\le \epsilon$. If the output hypothesis $h$ belongs to $C$, then we say that $C$ is {\it properly learnable} with $q$ queries. 

We now reduce the RCV problem to a learning problem.
\begin{lemma}\label{RCVforP}
    Let $C$ be a class of functions that is closed under zero-one projections. If $C$ is learnable from $k$-\textsc{Junta} with $Q(n,\epsilon,\delta)$ queries, then $C$ is $k$-relevant coordinates verifiable with $Q(n,\epsilon/(ck),\delta/2)+O(k+\log (1/\delta))$ queries.
\end{lemma}
\begin{proof} Let $A(n,\epsilon,\delta)$ be an algorithm that learns $C$ from $k$-\textsc{Junta} with $Q(n,\epsilon,\delta)$ queries.
    Consider the following algorithm describes in Algorithm~\ref{Alg2}. We run $A(n,\epsilon/(ck),\delta/2)$, and with probability at least $1-\delta/2$, it returns $g\in k$-\textsc{Junta} such that $\Pr_x[f(x)\not=g(x)]\le \epsilon/(ck)$. Let $U$ be the set of all indices $i$ where $x_i$ is a variable in $g$. Then $|U|\le k$. 
    Next, let $W=U$. For each $j\in W$, we estimate $\Pr[g_{|x_j\gets 0}(x)\not= g_{|x_j\gets 1}(x)]$ up to an additive error of $\epsilon/(ck)$ with confidence $1-\delta/(4k)$. 
    If the estimation is less than or equal to $5\epsilon/(ck)$, we remove $j$ from $W$ and move to the next index in $W$. 
    If the estimation is greater than $5\epsilon/(ck)$, we iterate $t:=(ck/\epsilon)\log(4k/\delta)$ times. At each iteration, we choose a random uniform $a$. If $g_{|x_j\gets 0}(a)\not= g_{|x_j\gets 1}(a)$, then we make two queries to check if $f_{|x_j\gets 0}(a)\not= f_{|x_j\gets 1}(a)$. 
    If this condition is satisfied,  we set $w^{(j)}=a$, stop iterating, and move to the next index in $W$. When all indices in 
$W$ are processed, the algorithm sets $V=W$.

\begin{algorithm}
\caption{Reduction to learning}\label{Alg2}
\begin{algorithmic}[1]
    \STATE $g\gets A(n, \epsilon/(ck), \delta/2)$.
    \STATE Let $U$ be the set of all $i$ where $x_i$ is a variable in $g$. 
    \STATE Initialize $W \gets U$.
    \FOR{each $j \in W$}
        \STATE Estimate $\Pr[g_{|x_j \gets 0}(x) \not= g_{|x_j \gets 1}(x)]$ up to additive error $\epsilon/(ck)$ with confidence $1-\delta/(4k)$.
        \IF{the estimation $\leq 5\epsilon/(ck)$}
            \STATE Remove $j$ from $W$.
        \ELSE
            \FOR{$t=(ck/\epsilon)\log(4k/\delta)$ iterations}
                \STATE Choose a random uniform $a$.
                \IF{$g_{|x_j \gets 0}(a) \not= g_{|x_j \gets 1}(a)$}
                    \STATE\label{jjjk}Make two queries to check if $f_{|x_j \gets 0}(a) \not= f_{|x_j \gets 1}(a)$.
                    \STATE \textbf{if} $f_{|x_j \gets 0}(a) \not= f_{|x_j \gets 1}(a)$ \textbf{then} set $w^{(j)} = a$;  stop the {\bf for} iteration.
                \ENDIF
            \ENDFOR
        \ENDIF
    \ENDFOR
    \STATE $V \gets W$.
\end{algorithmic}
\end{algorithm}

    We now prove
    \begin{claim}\label{Cl01} At any iteration, with probability at least $1-\delta/2-m\delta/(4k)$
     $$\Pr_{x,y}[f(x_W\circ y_{\overline{W}})\not=f(x)]\le \frac{(4m+2)\epsilon}{ck},$$
     where $m=|U\backslash W|$.
    \end{claim}
    \begin{proof}
    The proof proceeds by induction on $m$. 
Initially, $W=U$ and $m=0$. Since with probability at least $1-\delta/2$, $\Pr_x[f(x)\not=g(x)]\le \epsilon/(ck)$, and $g(x)$ is independent of the variables $x_i$, for $i\in \overline{U}$, it follows that $\Pr_{x,y}[f(x_U\circ y_{\overline{U}})\not=g(x)]\le \epsilon/(ck)$. Therefore, by the triangle inequality for probabilities, with probability at least $1-\delta/2$,
\begin{eqnarray*}
    \Pr_{x,y}[f(x_U\circ y_{\overline{U}})\not=f(x)]\le  \Pr_{x,y}[f(x_U\circ y_{\overline{U}})\not=g(x)]+\Pr_{x,y}[f(x)\not=g(x)]\le \frac{2\epsilon}{ck}.
\end{eqnarray*}
This establishes the base case for $m=0$. 

Assuming the hypothesis holds for $m$, we prove it for $m+1$. That is, with probability at least $1-\delta/2-m\delta/(4k)$, we have
     $\Pr_{x,y}[f(x_W\circ y_{\overline{W}})\not=f(x)]\le {(4m+2)\epsilon}/{(ck)}.$

The value of $|U\backslash W|$  increases from $m$ to $m+1$ when the algorithm removes some $j\in W$ from~$W$. 
In that case, the estimation of $\Pr[g_{|x_j\gets 0}(x)\not= g_{|x_j\gets 1}(x)]$ is less than or equal to $5\epsilon/(ck)$. Thus, with probability at least $1-\delta/(4k)$, we have $\Pr[g_{|x_j\gets 0}(x)\not= g_{|x_j\gets 1}(x)]\le 6\epsilon/(ck)$. 
Denote $g_\xi:=g_{|x_j\gets \xi}$ and $f_\xi:=f_{|x_j\gets \xi}$ for $\xi\in\{0,1\}$. Since 
$$\frac{\epsilon}{ck}\ge \Pr[f\not=g]=\frac{1}{2}\Pr[f_0\not=g_0]+\frac{1}{2}\Pr[f_1\not=g_1]$$ we have
\begin{eqnarray}\label{fZgZ}
    \Pr[f_0\not= g_0]+\Pr[f_1\not= g_1]\le \frac{2\epsilon}{ck}.
\end{eqnarray}
Using the triangle inequality for probabilities, we get
\begin{eqnarray*}
    \Pr[f_0\not=f_1]&\le&\Pr[f_0\not=g_0]+\Pr[f_1\not=g_1]+\Pr[g_0\not=g_1]
    \le \frac{8\epsilon}{ck}.
\end{eqnarray*}
Thus,
$$\Pr_{x,y}[f(x)\not=f_{|x_j\gets y_j}(x)]=\frac{1}{2}\Pr[f_0\not=f_1]\le \frac{4\epsilon}{ck}.$$
By Lemma~\ref{InfLemma} and the induction hypothesis, with probability at least $1-\delta/2-(m+1)\delta/(4k)$, we have
\begin{eqnarray*}
    \Pr_{x,y}[f(x_{W\backslash \{j\}}\circ y_{\overline{W}\cup\{j\}})\not=f(x)]&\le& \Pr_{x,y}[f(x_{W}\circ y_{\overline{W}})\not=f(x)]+\\&&\Pr_{x,y}[f_{|x_j\gets y_j}(x)\not=f(x)]\\
    &\le& \frac{(4m+2)\epsilon}{ck}+\frac{4\epsilon}{ck}=\frac{(4(m+1)+2)\epsilon}{ck}.
\end{eqnarray*}
This completes the proof of Claim~\ref{Cl01}
\end{proof}
Now, by Claim~\ref{Cl01} and since $m\le |U|\le k$ we have, with probability at least $1-3\delta/4$,
$$\Pr_{x,y}[f(x_{V}\circ y_{\overline{V}})\not=f(x)]\le \frac{(4k+2)\epsilon}{ck}\le \epsilon.$$
This proves item~\ref{Trc2}.

To prove item~\ref{Trc3}, we prove the following.  
\begin{claim}
 With probability at least $1-\delta/8$, all variables in $V$ have witnesses.   
\end{claim}
\begin{proof}
Let $j\in V$. Then the estimation of $\Pr[g_{|x_j\gets 0}(x)\not= g_{|x_j\gets 1}(x)]$ is greater than $5\epsilon/(ck)$, and with probability at least $1-\delta/(4k)$, we have
$$\Pr[g_{|x_j\gets 0}(x)\not= g_{|x_j\gets 1}(x)]\ge \frac{4\epsilon}{ck}.$$
Now, using the triangle inequality for probabilities and (\ref{fZgZ}),
\begin{eqnarray}
\Pr[f_0\not=f_1|g_0\not=g_1]&\ge&\Pr[g_0\not=g_1|g_0\not=g_1]-\\&&\Pr[g_0\not=f_0|g_0\not=g_1]-\Pr[g_1\not=f_1|g_0\not=g_1]\nonumber\\
    &\ge&1-\frac{\Pr[g_0\not=f_0]}{\Pr[g_0\not=g_1]}-\frac{\Pr[g_1\not=f_1]}{\Pr[g_0\not=g_1]}\nonumber\\
     &\ge&1-\frac{\Pr[g_0\not=f_0]+\Pr[g_1\not=f_1]}{\Pr[g_0\not=g_1]}\nonumber\\
     &\ge& 1-\frac{2\epsilon/(ck)}{4\epsilon/(ck)}=\frac{1}{2}.\label{half23}
\end{eqnarray}
Therefore,
\begin{eqnarray*}
    \Pr[(f_0\not=f_1)\wedge (g_0\not=g_1)]=\Pr[f_0\not=f_1|g_0\not=g_1]\Pr[g_0\not=g_1]\ge \frac{1}{2}\cdot\frac{4\epsilon}{ck}=\frac{2\epsilon}{ck}.
\end{eqnarray*}
The probability that after $t$ iterations, the algorithm cannot find $a$ such that $g_0(a)\not= g_1(a)\wedge f_0(a)\not=f_1(a)$ is at most
$$\left(1-\frac{2\epsilon}{ck}\right)^t\le \frac{\delta}{8k}.$$
Thus, after $t$ iterations, with probability at least $1-\delta/(8k)$, we obtain an $a$ such that $(g_0(a)\not= g_1(a))\wedge (f_0(a)\not=f_1(a))$, which serves as a witness for $x_j$. Consequently, with probability at least $1-\delta/8$, all variables in $V$ have witnesses.
\end{proof}
Finally, we analyze the query complexity and show that it is equal to $Q(n,\epsilon/(ck),$ $\delta/2)+O(k+\log(1/\delta))$. The above algorithm runs $A(n,\epsilon/(ck),\delta/2)$, which has query complexity $Q(n,\epsilon/(ck),\delta/2)$. The algorithm then makes two queries $f_{|x_j\gets 0}(a)$ and $f_{|x_j\gets 1}(a)$ if $g_{|x_j\gets 0}(a)\not= g_{|x_j\gets 1}(a)$, and if they are not equal, it moves to the next index of $W$. Since $|W|\le k$, and by (\ref{half23}), $\Pr[f_0\not= f_1|g_0\not=g_1]=1/2$,
the expected number of queries made in step~\ref{jjjk} is $O(k)$. By Chernoff's bound, the algorithm can limit the number of these queries to $O(k+\log(1/\delta))$, and the failure probability is at most $\delta/8$. Thus, the total query complexity is $Q(n,\epsilon/(ck),\delta/2)+O(k+\log(1/\delta))$.
\end{proof}

We now prove
\begin{lemma}\label{RCVkJunta}
    Let $C\subseteq k$-$\textsc{Junta}$. Then $C$ is $k$-relevant coordinates verifiable with $O(k/\epsilon+k\log n)$ queries.
\end{lemma}
\begin{algorithm}
\caption{RC-Verify$(f,n,k,\epsilon)$}
\label{Verifier1}
\begin{algorithmic}[1]
    \STATE $V\gets \emptyset$
    \FOR{$O(k/\epsilon)$ times}\label{ForA3}
        \STATE\label{Draw3} Draw random uniform $a,b\in\{0,1\}^n$
        \STATE {\bf if} {$f(a_V\circ b_{\overline{V}})\not=f(a)$}
            {\bf then} 
        \STATE \hspace{.3in} Binary-Search$(f,a,a_V\circ b_{\overline{V}},V)$.
        \STATE \hspace{.3in} Let $j$ be the relevant coordinate and $w^{(j)}$ be the witness found by Binary-Search.
        \STATE \hspace{.3in} $V\gets V\cup\{j\}; W\gets W\cup\{w^{(j)}\}$
    \ENDFOR
    \STATE Output($V,W$).
\end{algorithmic}
\end{algorithm}
\begin{proof} The Binary-Search procedure takes two assignments $a$ and $b$ and a set $V\subseteq [n]$ such that $f(a_V\circ b_{\overline{V}})\not=f(a)$. If $|\overline{V}|=1$, then $\overline{V}=\{j\}$ for some $j\in [n]$, and $a$ is a witness for coordinate $j$. If $|\overline{V}|\ge 2$, then it splits $\overline{V}$ into two disjoint sets $\overline{V}=W_1\cup W_2$ of sizes that differ by at most~$1$, then evaluate $f(a_{V\cup W_1}\circ b_{W_2})$. Then, either $f(a)\not=f(a_{V\cup W_1}\circ b_{W_2})$, and we continue by recursively splitting $W_2$ by calling Binary-Search$(f,a,a_{V\cup W_1}\circ b_{W_2},W_2)$, or $f(a_V\circ b_{W_2}\circ a_{W_1})=f(a_{V\cup W_1}\circ b_{W_2})\not= f(a_V\circ b_{\overline{V}})= f(a_V\circ b_{W_2}\circ b_{W_1})$, and we continue by recursively splitting $W_1$ by calling Binary-Search$(f,a_V\circ b_{W_2}\circ a_{W_1},a_V\circ b_{W_2}\circ b_{W_1},W_1)$. It is evident that the query complexity of this procedure is $O(\log n)$.

\begin{algorithm}
\caption{Binary-Search$(f,a,a_V\circ b_{\overline{V}},V)$}
\begin{algorithmic}[1]
\REQUIRE Assignments \( a \), \( a_V\circ b_{\overline{V}} \), a set \( V \subseteq [n] \), such that \( f(a_V \circ b_{\overline{V}}) \neq f(a) \).
\ENSURE Relevant coordinate \( j\in\overline{V} \) and witness \( w^{(j)} = a \).

\IF{\( |\overline{V}|=1\  \AND \ \overline{V} = \{j\} \)}
    \STATE Return \( j \) and \( w^{(j)} = a \).
\ENDIF
\STATE Split \( \overline{V} \) into two disjoint sets \( \overline{V} = W_1 \cup W_2 \), where \( |W_1| \) and \( |W_2| \) differ by at most 1.
\IF{\( f(a) \neq f(a_{V \cup W_1} \circ b_{W_2}) \)}
    \STATE \textsc{Binary-Search}\((f, a, a_{V \cup W_1} \circ b_{W_2}, W_2)\).
\ELSE
\STATE {\textsc{Binary-Search}\((f, a_V \circ b_{W_2} \circ a_{W_1}, a_V \circ b_{W_2} \circ b_{W_1}, W_1)\).}
\ENDIF
\end{algorithmic}
\end{algorithm}

First, note that the Binary-Search procedure is executed only when $f(a_V\circ b_{\overline{V}})\neq f(a)$, where $V$ represents the set of relevant coordinates discovered so far. Therefore, each time the algorithm executes Binary-Search, it finds a new relevant coordinate. Since the query complexity of Binary-Search is $\log n$ and $C\subseteq k$-$\textsc{Junta}$, the total number of queries made by Binary-Search is at most $k\log n$. 

   By Lemma~\ref{InfLemma}, for $V=U\cup \{i\}$, we have $$\Pr_{x,y}[f(x_U\circ y_{\overline{U}})\not=f(x)]=\frac{1}{2}\Inf_f(\overline{U})\ge \frac{1}{2}\Inf_f(\overline{V})= \Pr_{x,y}[f(x_V\circ y_{\overline{V}})\not=f(x)],$$
   and therefore, if $\Pr_{x,y}[f(x_V\circ y_{\overline{V}})\not=f(x)]>\epsilon$, then $\Pr_{x,y}[f(x_U\circ y_{\overline{U}})\not=f(x)]>\epsilon$. 
   Hence, the probability that the algorithm fails to output $V$ such that $\Pr_{x,y}[f(x_V\circ y_{\overline{V}})\not=f(x)]\le \epsilon$ is less than the probability that for $O(k/\epsilon)$ Bernoulli trials with a success probability $\epsilon$, fewer than $k$ success occur. By Chernoff's bound, this probability is less than any constant $\eta$. 
   
  Thus, with probability at least $1-\eta$, the final $V$ satisfies $\Pr_{x,y}[f(x_V\circ y_{\overline{V}})\not=f(x)]\le \epsilon$.
\end{proof}

We now prove a similar result for classes that are subsets of $s$-\textsc{Term Function}.
\begin{lemma}\label{stf02}
Let $C\subset$ $s$-\textsc{Term Function}. The class $C$ is $O(s\log(s/\epsilon))$-relevant coordinates verifiable with $O((1/\epsilon+\log n)s\log(s/\epsilon))$  queries.
\end{lemma}
\begin{proof} We run {\bf Algorithm~\ref{Verifier1}} RC-Verify$(f,n,k,\epsilon/3)$ with $k=5\log (s/\epsilon)$.

    Let $F=f(T_1,T_2,\ldots,T_s)$ be the target function, where $f:\{0,1\}^s\to \{0,1\}$ and $T_1,T_2,\ldots,T_s$ are any terms. Suppose without loss of generality that $|T_1|\le |T_2|\le \cdots \le |T_{s'}|\le k<|T_{s'+1}|<\cdots < |T_s|$. Let $F'=f(T_1,\ldots,T_{s'},0,\ldots,0)$. 

    First, we have 
    \begin{eqnarray}\label{knr034}
       \Pr_x[F(x)\not=F'(x)]\le \Pr_x[(\exists i>s') T_i(x)=1]\le s 2^{-5\log(s/\epsilon)}= \epsilon^5/s^4\le \epsilon/3 
    \end{eqnarray}
    and $F'\in (ks)$-$\textsc{Junta}$. 

    For two assignments $a'$ and $b'$, we define $[a',b']$ the vector $y$ that satisfies $y_i=a'_i$ if $a'_i=b'_i$, and $y_i=x_i$ if $a'_i\not=b'_i$. 
    
    Now, for $t=O(k/\epsilon)$ random uniform $a^{(i)},b^{(i)}$ drawn in step~\ref{Draw3} of {\bf Algorithm~\ref{Verifier1}}, the probability that for some $i\in [t]$, one of the terms $T_j$, $j>s'$, satisfies $T_j([a^{(i)}_V\circ b^{(i)}_{\overline V},a^{(i)}])\not\equiv 0$ is at most 
    $$2ts\left(\frac{3}{4}\right)^{k}=O\left(\frac{\log(s/\epsilon)}{\epsilon}s\left(\frac{3}{4}\right)^{5\log(s/\epsilon)}\right)=o(1).$$
    It is also clear that, for any term $T$, if $T([a_V\circ b_{\overline V},a])=0$, then all the assignments $v$ created in the binary search Binary-Search$(f,a,a_V\circ b_{\overline{V}},V)$ satisfy $T(v)=0$. 
    Therefore, with probability $1-o(1)$, all the queries $v$ in  {\bf Algorithm~\ref{Verifier1}} satisfy $F(v)=F'(v)$. 
    
    By Lemma~\ref{RCVkJunta}, {\bf Algorithm~\ref{Verifier1}}, with probability at least $1-\eta$, returns $V$ and $w^{(j)}$, $j\in V$, that satisfy conditions~\ref{Trc1}-\ref{Trc3} (with $\epsilon/3$) for $F'$. Using (\ref{knr034}) and condition~\ref{Trc2} for $F'$ (with $\epsilon/3$), we conclude that, with probability at least $1-\eta-o(1)$
    \begin{eqnarray*}
        \Pr_{x,y}[F(x_V\circ y_{\overline{V}})\not=F(x)]&\le& \Pr_{x,y}[F(x_V\circ y_{\overline{V}})\not=F'(x_V\circ y_{\overline{V}})]+ \\&&\Pr_{x,y}[F'(x_V\circ y_{\overline{V}})\not=F'(x)]+\\
        &&\Pr_{x,y}[F'(x)\not=F(x)]\\
        &\le& \epsilon.
    \end{eqnarray*} 
    This implies condition~\ref{Trc2} for $F$.
Since, with probability $1-o(1)$, all the assignments $v$ in the algorithm satisfy $F(v)=F'(v)$ and $F'$ is $(ks)$-junta, conditions~\ref{Trc3} and~\ref{Trc1} are also satisfied for~$F$. 
\end{proof}

We now prove
\begin{lemma}\label{MU}
    Let $C\subseteq k$-$\textsc{Junta}$. Let\footnote{The algorithm requires $\mu(C)$ only to determine its query complexity. Since the class $C$ is fixed and is not part of the input, $\mu(C)$, or any lower bound on it, can be treated as known.} $$\mu(C):=\min_{f\in C}\min_{i\in \RC(f)}\Pr_x[f_{|x_i\gets 0}(x)\not=f_{|x_i\gets 1}(x)]$$ 
    where $\RC(f)$ is the set of relevant coordinates in $f$. Then $C$ is $k$-relevant coordinates verifiable with $O(k/\mu(C)+k\log n)$ queries.
\end{lemma}
\begin{proof}
    We use Algorithm~\ref{Verifier1} and modify step~\ref{ForA3} to ``{\bf for} $O(k/\mu(C))$ times {\bf do}''. By Lemma~\ref{InfLemma}, the minimum value of $\Pr_{a,b}[f(a_V\circ b_{\overline{V}})\not=f(a)]$ occur when $|\overline{V}|=1$. Therefore,
    \begin{eqnarray*}
        \min_{f\in C}\min_{\overline{V}\subseteq \RC(f)} \Pr_{a,b}[f(a_V\circ b_{\overline{V}})\not=f(a)]&=& \min_{f\in C}\min_{\ i\in \RC(f)} \Pr_{x,y}[f(x)\not=f_{|x_i\gets y}(x)]\\
        &=& \frac{1}{2}\min_{f\in C}\min_{i\in \RC(f)}\\&&\Pr_x[f_{|x_i\gets 0}(x)\not=f_{|x_i\gets 1}(x)]\\
        &=&\frac{\mu(C)}{2}.
    \end{eqnarray*}
    Following the same steps as in the proof of Lemma~\ref{RCVkJunta}, the result follows. 
\end{proof}

\section{Blocks Verifiers}\label{BlocksVerifiers}
We now prove the following. 
\begin{lemma}\label{krcvKk}
    Let $K\ge k$ be two integers. Let $C\subseteq K$-$\textsc{Junta}$ be a class that is closed under variable permutations and zero-one projection. Suppose $C$ is $k$-relevant coordinates verifiable with $q(n,\epsilon)$ queries. Then $C$ is $(k,\alpha)$-relevant blocks verifiable with $q(O(K^2),O(\epsilon))+O(1/\epsilon+k/\alpha)$ queries. 
\end{lemma}
\begin{proof} Let RC-Verify$(f,n,k,\epsilon)$ be an RC-verifier for $C$ with $q(n,\epsilon)$ queries. Recall that for variables $y_1,y_2,\ldots,y_m$ and a partition $Y_1\cup Y_2\cup \cdots \cup Y_m=[n]$, the vector $z=y_1^{Y_1}\circ y_2^{Y_2}\circ\cdots \circ y_m^{Y_m}$ is the vector where $z_i=y_j$ if $i\in Y_j$. 
    Consider the algorithm RB-Verify$(f,n,K,k,\alpha,\epsilon)$ shown in Algorithm~\ref{alg:Generates}.
\begin{algorithm}
\caption{RB-Verify$(f,n,K,k,\alpha,\epsilon)$}
\label{alg:Generates}
\begin{algorithmic}[1]
    \STATE Randomly uniformly partition $[n]$ into $m:=K^2/\eta$ disjoint sets $Y_1,Y_2,\ldots,Y_m$.\label{Gen01}
    \STATE Run the verifier RC-Verify$(f',m,k,\eta\epsilon/4)$ on $f'(y):=f(y_1^{Y_1}\circ\ldots\circ y_m^{Y_m})$, where $y=(y_1,\ldots,y_m)$ are variables.\label{Gen02}
    \STATE\label{Thebs} Let $V=\{i_1,\ldots,i_{k'}\}$ be the output, $b^{(1)},\ldots,b^{(k')}\in \{0,1\}^m$ be the corresponding witnesses, and $u'\in \{0,1\}^m$ be a random uniform vector.\label{Gen03}
    \STATE\label{uPrime} Let $u=(u'_1)^{Y_1}\circ \cdots \circ (u'_{m})^{Y_m}$.
    \STATE\label{XjYij} Let $X_j=Y_{i_j}$ for all $j\in [k']$ and $X=\cup_{j=1}^{k'}X_j$.
            \STATE {\bf If} {{\bf Distinguish}$(\Pr_x[f(x_X\circ u_{\overline{X}})\neq f(x)], \frac{\epsilon}{4}, {\epsilon}) = 1$}\label{Gen06}
   {\bf then}\label{Reject1A1}  Reject.  \label{abc2}  
    \FOR{$j = 1$ to $k'$}
        \STATE\label{Theaj} Let $a^{(j)}=(b^{(j)}_1)^{Y_1}\circ (b^{(j)}_2)^{Y_2}\circ \cdots\circ (b^{(j)}_{m})^{Y_{m}}$.
        \STATE Run the testing algorithm {\bf TestLiteral} in Lemma~\ref{TestLiteral}, with success probability $1-\eta$ to test whether $f(a^{(j)}_{[n]\backslash X_j}\circ x_{X_j})$ is a literal or $\alpha$-far from any literal.\label{Gen09}
        \\ {\bf If} the testing algorithm rejects, {\bf then} Reject.\label{Gen11}
    \ENDFOR
    \STATE Output $X_1,\ldots,X_{k'}$, $a^{(1)},\ldots,a^{(k')}$, and $u$.
\end{algorithmic}
\end{algorithm}

If $f\not\in C$, the probability that the algorithm does not reject and the output does not satisfy conditions \ref{Trb1} and \ref{Trb2} is the probability that the algorithm does not reject in step~\ref{Gen06} while $\Pr_x[f(x_X\circ u_{\overline{X}})\not=f(x)]>\epsilon$, or does not reject in step~\ref{Gen11} while for some $j\in [k']$, $f(a^{(j)}_{[n]\backslash X_j}\circ x_{X_j})$ is $\alpha$-far from any literal. 
By Lemma~\ref{distinguish}, the probability of the former is at most $\eta$, and by Lemma~\ref{TestLiteral}, the probability of the latter is at most $\eta$. Therefore, if $f\not\in C$ and the algorithm does not reject, then, with probability at least $1-2\eta$, conditions \ref{Trb1} and \ref{Trb2} occur. 

Let $f\in C$. We now show that, with high probability, the algorithm does not reject, and items~\ref{Trb1} and~\ref{Trb3} (and therefore, also~\ref{Trb2}) hold.  

Since $C\subseteq K$-$\textsc{Junta}$, $f$ depends on at most $K$ coordinates. Consider step~\ref{Gen01} in the algorithm. With probability, at least
$$\left(1-\frac{1}{m}\right)\left(1-\frac{2}{m}\right)\cdots\left(1-\frac{K-1}{m}\right)\ge 1-\frac{K(K-1)}{2m}\ge 1-\frac{\eta}{2},$$
every $Y_i$ contains at most one relevant coordinate. 

Consider steps~\ref{Gen02} and~\ref{Gen03}. Now assume that every $Y_i$ contains at most one relevant coordinate. Based on this assumption, considering the symmetry of $C$ and noting that $f$ belongs to $C$, we have $f'(y)\in C$ with $m$ variables. By the definition of $k$-relevant coordinates verifiable, the RC-verifier RC-Verify$(f',m,k,\epsilon\eta/4)$, with probability at least $1-\eta$, returns $V$ and $b^{(j)}\in\{0,1\}^m$, $j\in [|V|]$, such that
\begin{enumerate}
\item $k':=|V|\le k.$
    \item \label{IO1}$\Pr_{y,z}[f'(y_V\circ z_{\overline{V}})\not=f'(y)]\le \eta\epsilon/4$ where $\overline{V}=[m]\backslash V$.
    \item For every $j\in [k']$, $b^{(j)}$ is a witness for the relevant coordinate $j$ in $f'(y)$. 
\end{enumerate}
Since $\Pr_{y,z}[f'(y_V\circ z_{\overline{V}})\not=f'(y)]\le \eta\epsilon/4$, by Markov's bound, for a random uniform $u'\in\{0,1\}^m$, with probability at least $1-\eta$, we have
\begin{eqnarray}
    \Pr_y[f'(y_V\circ u'_{\overline{V}})\not=f'(y)]\le \epsilon/4.
\end{eqnarray}

Consider $u=(u'_1)^{Y_1}\circ \cdots \circ (u'_{m})^{Y_m}$ defined in step~\ref{uPrime}, $X_j=Y_{i_j}$, and $X=\cup_{j=1}^{k'}X_j$ in step~\ref{XjYij}. Define \( x' = (x_{\tau_1}, \ldots, x_{\tau_m}) \), where for every $i\in [m]$, \( \tau_i \) is the relevant coordinate in \( f(x) \) in \( Y_i \), if such a coordinate exists, and an arbitrary coordinate in \( Y_i \) otherwise. Since $f'(x')=f(x)$, $f'(x_V'\circ u'_{\overline{V}})=f(x_X\circ u_{\overline{X}})$ and $\Pr_y[f'(y)\not=f'(y_V\circ u'_{\overline{V}})]\le \epsilon/4$, we have $\Pr_x[f'(x')\not=f'(x'_V\circ u'_{\overline{V}})]\le \epsilon/4$, and therefore $\Pr_x[f(x)\not=f(x_X\circ u_{\overline{X}})]\le \epsilon/4$. Thus, the algorithm, with probability at least $1-\eta$, does not reject in step~\ref{abc2}. This also implies condition~\ref{Trb1}.

We now show that $G_j(x):=f(a^{(j)}_{[n]\backslash X_j}\circ x_{X_j})\in \{x_{\tau_j},\overline{x_{\tau_j}}\}$, which implies that the algorithm does not reject in step~\ref{Gen11} and condition~\ref{Trb3} occurs. To this end, since there is only one relevant coordinate in $X_j$, we have $G_j(x)\in \{x_{\tau_j},\overline{x_{\tau_j}},0,1\}$. 

Consider steps~\ref{Thebs} and~\ref{Theaj}. Since $G_j(a^{(j)})=f(a^{(j)}_{[n]\backslash X_j}\circ a^{(j)}_{X_j})=f(a^{(j)})=f'(b^{(j)})$, and $b^{(j)}$ is a witness for $i_j\in V$, if we flip coordinate $i_j$ in $b^{(j)}$, we get $w$ that satisfies $f'(w)\not= f'(b^{(j)})$. Now, since $X_j=Y_{i_j}$, we have $f'(w)=f(a^{(j)}_{[n]\backslash X_j}\circ \overline{a^{(j)}_{X_j}})$, and therefore, 
$G_j(a^{(j)})=f(a^{(j)}_{[n]\backslash X_j}\circ a^{(j)}_{X_j})\not = f(a^{(j)}_{[n]\backslash X_j}\circ \overline{a^{(j)}_{X_j}})=G_j(\overline{a^{(j)}})$. Thus, $G_j$ cannot be a constant function.

This concludes the proof of the algorithm's correctness.

For the query complexity of the algorithm, we have the following: In step~\ref{Gen03}, it makes $q(m,\eta\epsilon/4)=q(O(K^2),O(\epsilon))$ queries. In step~\ref{Gen06}, by Lemma~\ref{distinguish}, it makes $O(1/\epsilon)$ queries. In step~\ref{Gen09}, by Lemma~\ref{TestLiteral}, it makes at most $O(k/\alpha)$ queries. 
\end{proof}

We now prove the following results.

\begin{lemma}\label{ELtoRBV}
    Let $C\subseteq k$-$\textsc{Junta}$ be a class that is closed under variable permutations and zero-one projection. If $C$ is exactly learnable with $Q(n)$ queries, then $C$ is $(k,\alpha)$-relevant blocks verifiable in
    $$Q(O(k^2))+O\left(\frac{1}{\epsilon}+\frac{k}{\alpha}\right)$$
    queries. 
\end{lemma}
\begin{proof}
    By Lemma~\ref{ELtoRCV}, $C$ is $k$-relevant coordinates verifiable with $Q(n)$ queries. Then the result follows by Lemma~\ref{krcvKk} with $K=k$. 
\end{proof}

\begin{lemma}\label{Imp1}
    Let $C\subseteq k$-$\textsc{Junta}$ be a class that is closed under variable permutations and zero-one projection. Then $C$ is $(k,\alpha)$-relevant blocks verifiable with $$O\left(\frac{k}{\epsilon}+k\log k+\frac{k}{\alpha}\right)$$ queries. 
\end{lemma}
\begin{proof}
    This Lemma follows from Lemma~\ref{RCVkJunta}, and Lemma~\ref{krcvKk} with $K=k$.
\end{proof}

\begin{lemma}\label{Imp12}
    Let $C\subseteq k$-$\textsc{Junta}$ be a class that is closed under variable permutations and zero-one projection. Then $C$ is $(k,\alpha)$-relevant blocks verifiable with $$O\left(\frac{k}{\mu(C)}+k\log k+\frac{k}{\alpha}+\frac{1}{\epsilon}\right)=O\left({k}{2^k}+\frac{k}{\alpha}+\frac{1}{\epsilon}\right)$$ queries. 
\end{lemma}
\begin{proof}
    The result follows from Lemma~\ref{MU} and Lemma~\ref{krcvKk} with $K=k$, along with the fact that for any nonzero $k$-junta $f$, we have $\Pr_x[f(x)=1]\ge {1}/{2^k}$.
\end{proof}

\begin{lemma}\label{Imp2}
     Let $K=\tilde O(s^2)$. Let $C\subseteq C':=s$-\textsc{Term Function}$\cap K$-\textsc{Junta} be a class that is closed under variable permutations and zero-one projection. 
     Then $C$ is $(O(s\log(s/\epsilon)),\alpha)$-relevant blocks verifiable with $$O\left(s\log\frac{s}{\epsilon}\left(\frac{1}{\epsilon}+\log s+\frac{1}{\alpha}\right)\right)$$ queries. 
\end{lemma}

\begin{proof} 
    By Lemma~\ref{stf02}, $C'$ is $O(s\log(s/\epsilon))$-relevant coordinates verifier with $O(s$ $\log(s/\epsilon)/\epsilon+s\log(s/\epsilon)\log n)$ queries. By Lemma~\ref{krcvKk}, $C'$ is $(O(s\log(s/\epsilon)),\alpha)$-relevant block verifier in 
    $$O\left(s\log\frac{s}{\epsilon}\left(\frac{1}{\epsilon}+\log s+\frac{1}{\alpha}\right)\right)$$
    queries. 
\end{proof}

\section{Self Corrector}\label{SelfCorrector}
Consider the procedure in Algorithm~\ref{alg:D} where $\oplus$ denotes the exclusive or. 
\begin{algorithm}
\caption{SelfCorrect$(G(x),a,\alpha,\delta)$.}
\label{alg:D}
\begin{algorithmic}[1]
    \STATE Draw uniformly at random $t=O(\max(1,\log(1/\delta)/\log(1/\alpha))$ assignments $u^{(1)},\ldots,u^{(t)}\in \{0,1\}^n$.
    \STATE Return $\Majority_j(G(u^{(j)})\oplus G(u^{(j)}\oplus a))$
\end{algorithmic}
\end{algorithm}

\begin{lemma} \label{selfcorrector} 
Let 
$$t=O\left(\max\left(1,\frac{\log\frac{1}{\delta}}{\log\frac{1}{\alpha}}\right)\right).$$
If $G(x)$ is $\alpha$-close to $\{x_i,\overline{x_i}\}$, then for any assignment $a\in \{0,1\}^n$, 
$$\Pr_{u^{(1)},\ldots,u^{(t)}\sim \{0,1\}^n}\left[\mbox{\rm SelfCorrect}(G(x),a,\alpha,\delta)=a_i\right]\ge 1-\delta$$
and if $G(x)\in\{x_i,\overline{x_i}\}$, then for any assignment $a\in \{0,1\}^n$, 
$$\Pr_{u^{(1)},\ldots,u^{(t)}\sim \{0,1\}^n}\left[\mbox{\rm SelfCorrect}(G(x),a,\alpha,\delta)=a_i\right]= 1.$$
\end{lemma}
\begin{proof}
    If $G(x)$ is $\alpha$-close to $x_i\oplus \xi$ for some $\xi\in\{0,1\}$, then for a random uniform $u^{(j)}$, with probability at least $1-2\alpha$, $G(u^{(j)})\oplus G(u^{(j)}\oplus a)=(u^{(j)}_i\oplus \xi)\oplus (u^{(j)}_i\oplus a_i\oplus \xi)=a_i$. 

    By Chernoff's bound the result follows.
\end{proof}

\section{Results}\label{Results}
In this section we give all the results of this paper mentioned in the introduction.

We start with some general results.
\begin{lemma}\label{Result1}
Let $C\subseteq$\textsc{$k$-Junta} be a class that is closed under variable permutations and zero-one projection. Suppose $C[k]$ is testable with $Q(k,\epsilon)$ queries. There is a testing algorithm for $C$ with $$Q(k,\epsilon/4)+O\left(\frac{k}{\epsilon}+k\log k\right)$$ queries.
\end{lemma}
\begin{proof}
    The result follows from Lemma~\ref{Imp1} and Lemma~\ref{Gresult} when $K=k$ and $\alpha=\epsilon$.
\end{proof}

\begin{lemma}\label{Result2}
Let $C\subseteq$\textsc{$k$-Junta} be a class that is closed under variable permutations and zero-one projection. Suppose $C[k]$ is testable with $Q(k,\epsilon)$ queries. Then there is a testing algorithm for $C$ with $$Q(k,\epsilon/4)+O\left(\frac{k}{\mu(C)}+\frac{k\log^2 k}{\log\log k}+\frac{1}{\epsilon}\right)$$ queries.
\end{lemma}
\begin{proof}
    By Lemma~\ref{Gresult} with $K=k$ and~\ref{Imp12}, there is a testing algorithm for $C$ with
    $$Q(k,\epsilon/4)+O\left(k\frac{\log \frac{1}{\epsilon}(\log k+\log\log \frac{1}{\epsilon})}{\log\frac{1}{\alpha}}+\frac{1}{\epsilon}\right)+O\left(\frac{k}{\mu(C)}+k\log k+\frac{k}{\alpha}+\frac{1}{\epsilon}\right)$$ queries.
    Setting $$\alpha=\frac{\log\log k+\log\log\frac{1}{\epsilon}}{\log k\log \frac{1}{\epsilon}}$$ we obtain
    $$Q(k,\epsilon/4)+O\left(k\frac{\log \frac{1}{\epsilon}(\log k+\log\log \frac{1}{\epsilon})}{\log\log k +\log\log\frac{1}{\epsilon}}+\frac{1}{\epsilon}+\frac{k}{\mu(C)}+k\log k\right).$$
    Now, it suffices to show that
    \begin{eqnarray}
        \label{Cl4}
    k\frac{\log \frac{1}{\epsilon}(\log k+\log\log \frac{1}{\epsilon})}{\log\log k +\log\log\frac{1}{\epsilon}}+\frac{1}{\epsilon}=O\left(\frac{k\log^2k}{\log\log k}+\frac{1}{\epsilon}\right).
    \end{eqnarray}
    
    Now we have two cases. If $k\log k\le 1/(\epsilon\log (1/\epsilon))$, then
    \begin{eqnarray*}
      k\frac{\log \frac{1}{\epsilon}(\log k+\log\log \frac{1}{\epsilon})}{\log\log k +\log\log\frac{1}{\epsilon}}+\frac{1}{\epsilon}&=&\frac{\log \frac{1}{\epsilon}(k\log k+k\log\log \frac{1}{\epsilon})}{\log\log k +\log\log\frac{1}{\epsilon}}+\frac{1}{\epsilon}\\
      &\le & \frac{\log\frac{1}{\epsilon}\left({\frac{1}{\epsilon\log\frac{1}{\epsilon}}}+{\frac{1}{\epsilon\log\frac{1}{\epsilon}}}\log\log\frac{1}{\epsilon}\right)}{\log\log\frac{1}{\epsilon}}+\frac{1}{\epsilon}\\
      &=&O\left(\frac{1}{\epsilon}\right).
    \end{eqnarray*}
    The other case is when $k\log k> 1/(\epsilon\log (1/\epsilon))$. Since $\log\log x<(1/2)\log x$ we get
    $$(3/2)\log k>\log k+\log\log k>\log\frac{1}{\epsilon}-\log\log\frac{1}{\epsilon}>(1/2)\log\frac{1}{\epsilon}$$
    and therefore $3\log k>\log(1/\epsilon)$.
Now
    \begin{eqnarray*}
      k\frac{\log \frac{1}{\epsilon}(\log k+\log\log \frac{1}{\epsilon})}{\log\log k +\log\log\frac{1}{\epsilon}}+\frac{1}{\epsilon}&\le &k\frac{3\log k(\log k+\log(3\log k))}{\log\log k }+\frac{1}{\epsilon}\\
      &=&O\left(\frac{k\log^2k}{\log\log k}+\frac{1}{\epsilon}\right).
    \end{eqnarray*}
    This completes the proof. 
\end{proof}

Since $\mu($\textsc{$k$-Junta}$)=2^{-k}$, by Lemma~\ref{Result2} we have.
\begin{lemma}\label{Result3}
Let $C\subseteq$\textsc{$k$-Junta} be a class that is closed under variable permutations and zero-one projection. Suppose $C[k]$ is testable with $Q(k,\epsilon)$ queries. Then there is a testing algorithm for $C$ with $$Q(k,\epsilon/4)+O\left(k2^k+\frac{1}{\epsilon}\right)$$ queries.
\end{lemma}

We now prove.
\begin{lemma}\label{sTFT}
Let $C\subseteq$\textsc{$s$-Term Function} be a class that is closed under variable permutations and zero-one projection. Suppose $C[O(s\log(s/\epsilon))]$ is testable with $Q(s,\epsilon)$ queries. Then there is a testing algorithm for $C$ with $$m(\epsilon)=Q(s,\epsilon/8)+O\left(\left(\frac{s}{\epsilon}+s\log s\right)\log\frac{s}{\epsilon}\right).$$ queries.
\end{lemma}
\begin{proof}
Let $K=\tilde O(s^2)$, $k=O(s\log(s/\epsilon))$ and $\alpha=\epsilon$. By Lemma~\ref{Imp2}, $C\cap K$-\textsc{Junta} is $(k,\epsilon)$-relevant blocks verifiable with $O(s\log(s/\epsilon)(1/\epsilon+\log s))$ queries. By Lemma~\ref{Gresult}, there is a testing algorithm for $C\cap K$-\textsc{Junta} with $m(2\epsilon)$ queries. Then the result for $C$ follows from Lemma~\ref{sTF01}. 
\end{proof}

We now give the following 

We say that a class of functions $C$ is {\it properly learnable in time $T$} if there is a randomized algorithm that, given access to a black-box for a function $f\in C$, runs in time $T$ and, with probability at least $1-\eta$, outputs a function $g\in C$ that is $\epsilon$-close to $f$. We say that $C$ is {\it properly learnable} if it is properly learnable in polynomial time. 

\begin{definition}\label{Defini}Let $C$ be a class of functions, and let $A$ be a learning algorithm that properly learns $C$. We say that the pair $(C,A)$ is {\it membership testable} in time $t$ if there exists an algorithm $B$ such that, for any Boolean function $f$, $A(f)$ outputs $g$, and $B$ decides whether $g\in C$ in time $t$. Notice here that $g$ may not be a Boolean function. 
\end{definition} 

The following result is Proposition~3.1.1 in~\cite{GoldreichGR98}. 
\begin{lemma}\label{LearningToTesting}
  Let $C$ be a class of functions, and let $A(\epsilon)$ be a learning algorithm that properly learns $C$ in time $T(\epsilon)$ with $m(\epsilon)$ queries. If $(C,A)$ is membership testable in time $t$ then $C$ is testable in time $T(\epsilon/2)+O(n/\epsilon)+t$ with $m(\epsilon/2)+O(1/\epsilon)$ queries.
\end{lemma}

\noindent
{\it Proof Sketch.} Let $B$ be an algorithm that tests membership in $C$ for functions output by $A$. The following is a tester for $C$.

Run $A(\epsilon/2)$ on $f$ and let $g$ be the output. Then run $B(g)$. If $g\not\in C$, reject. Otherwise, test whether $f(a)=g(a)$ on $O(1/\epsilon)$ uniformly random inputs $a\in\{0,1\}^n$. If $f(a)\not=g(a)$ for any such $a$, reject; otherwise, accept.\qed

We now prove the following 
\begin{theorem}
    Let $C\subseteq$\textsc{$s$-Term Function} be a class that is closed under variable permutations and zero-one projection.  Then
    \begin{enumerate}
        \item\label{TH0101} If $C[O(s\log(s/\epsilon))]$ is properly learnable with $m(\epsilon)$ queries, then there is a testing algorithm for $C$ with $$m(\epsilon/16)+O\left(\left(\frac{s}{\epsilon}+s\log s\right)\log\frac{s}{\epsilon}\right)$$ queries.
        \item\label{TH0102} There is an exponential-time testing algorithm for $C$ with $$O\left(\frac{\log|C[O(s\log(s/\epsilon))]|}{\epsilon} +\left(\frac{s}{\epsilon}+s\log s\right)\log\frac{s}{\epsilon}\right).$$ queries.
        \item\label{TH0103} The classes\footnote{See the definition of these classes in ~\cite{DiakonikolasLMORSW07}} of $s$-Term DNF, size $s$-Decision Trees, size-$s$ Branching Programs and size-$s$ Boolean Formulas are testable in exponential time with $\tilde O(s/\epsilon)$ queries. The class of size-$s$ Boolean Circuits is testable in exponential time with $\tilde O(s^2/\epsilon)$ queries. 
        \item\label{TH0104} The classes $s$-Term Monotone DNF and $s$-Term Unate DNF are testable (in polynomial time) with $\tilde O(s/\epsilon)$ queries. 
    \end{enumerate}
\end{theorem}
\begin{proof}
    Item~\ref{TH0101} follows immediately from Lemma~\ref{sTFT} and Lemma~\ref{LearningToTesting}. 
    
    We now prove item~\ref{TH0102}. By Occam's Razor learning algorithm~\cite{BlumerEHW87}, $C[O(s\log(s/\epsilon))]$ is learnable in exponential time with $m=O(\log|C[O(s\log(s/\epsilon))]|/\epsilon)$ queries. The result then follows from item~\ref{TH0101}.

    Item~\ref{TH0103} follows from item~\ref{TH0102} and the fact that $|C[t]|=2^{\tilde O(t)}$ for the classes of $s$-Term DNF, size $s$-Decision Trees, size-$s$ Branching Programs, size-$s$ Boolean Formulas, and $|C[t]|=2^{\tilde O(t^2)}$ for the class of size-$s$ Boolean Circuits (see~\cite{DiakonikolasLMORSW07}).

    Item~\ref{TH0104} follows from item~\ref{TH0101} and the fact that both classes $s$-Term Monotone DNF and $s$-Term Unate DNF are properly learnable in polynomial time with $\tilde O(s/\epsilon)$ queries~\cite{Bshouty20x}.
\end{proof}

\subsection{Testing \textsc{$k$-Junta}}
For \textsc{$k$-Junta} in the uniform distribution framework,
Ficher et al.~\cite{FischerKRSS02} introduced the junta testing problem and presented a non-adaptive algorithm with $\tilde O(k^2)/\epsilon$ queries. Blais, in~\cite{Blais08}, presented a non-adaptive algorithm with $\tilde{O}(k^{3/2})/\epsilon$ queries, and in~\cite{Blais09}, he gave an adaptive algorithm with $O(k\log k+k/\epsilon)$ queries. The latter result also follows from Lemma~\ref{Result1}.\footnote{In Lemma~\ref{Result1}, for \textsc{$k$-Junta} we have $Q(k,\epsilon/4)=0$.}

On the lower bounds side, Fisher et al.~\cite{FischerKRSS02} presented an $\Omega(\sqrt{k})$ lower bound for non-adaptive testing. Chockler and Gutfreund~\cite{ChocklerG04} presented an $\Omega(k)$ lower bound for adaptive testing and, which was improved to $\Omega(k\log k)$ by Sa\u{g}lam in~ \cite{Saglam18}. The lower bound $\Omega(1/\epsilon)$ follows from~\cite{BshoutyG22,Eldar}. For non-adaptive testing Chen et al.~\cite{ChenSTWX17} presented the lower bound $\tilde{\Omega}(k^{3/2})/\epsilon$.

For testing \textsc{$k$-Junta} in the distribution-free model, Chen et al.~\cite{LiuCSSX18} presented a one-sided adaptive algorithm with $\tilde{O}(k^{2})/\epsilon$ queries and proved a lower bound $\Omega(2^{k/3})$ for any non-adaptive algorithm. The work of Halevy and Kushilevitz in~\cite{HalevyK07} gives a one-sided non-adaptive algorithm with $O(2^k/\epsilon)$ queries. The adaptive $\Omega(k\log k)$ uniform-distribution lower bound from~\cite{Saglam18} trivially extends to the distribution-free model.
Bshouty~\cite{Bshouty19,Bshouty20x} presented a two-sided adaptive algorithm with $\tilde O(1/\epsilon)k\log k$ queries. All these algorithms make at least $\Omega(k/\epsilon)$ queries 

Our algorithm in this paper gives.
\begin{lemma}
There is a testing algorithm for \textsc{$k$-Junta} with $$O\left(k2^k+\frac{1}{\epsilon}\right)$$ queries.
\end{lemma}
\begin{proof}
    The class \textsc{$k$-Junta}$[k]$ is testable with no queries (just output accept) since every function in \textsc{$k$-Junta}$[k]$ is $k$-junta. The result then follows directly from Lemma~\ref{Result3}.
\end{proof}

\subsection{Functions with Fourier Degree at most $d$}
For convenience, we take the Boolean functions to be $f:\{-1,1\}^n\to \{-1,1\}$. Then every Boolean function has a unique Fourier representation $$f(x)=\sum_{S\subseteq [n]} \hat f_S\chi_S(x)$$ where $\chi_s(x)=\prod_{i\in S}x_i$ and $\hat f_S$ are the {\it Fourier coefficients} of $f$. The {\it Fourier degree} of $f$ is defined as the largest $d=|S|$ such that $\hat f_S\not=0$.

Let \textsc{Ffd}$(d)$ denote the class of all Boolean functions over $\{-1,1\}^n$ with Fourier degree at most~$d$. Wellens~\cite{Wellens20} proved that any Boolean function in \textsc{Ffd}$(d)$ must have at most $k:=4.394\cdot 2^d=O(2^d)$ relevant variables. See also~\cite{ChiarelliHS20}. Diakinikolas et al.~\cite{DiakonikolasLMORSW07}, showed that every nonzero Fourier coefficient of a function  $f\in$ \textsc{Ffd}$(d)$ is an integer multiple of $1/2^{d-1}$. Since $\sum_{S\subseteq [n]} \hat f_S^2=1$, there are at most $2^{2d-2}$ nonzero Fourier coefficients in any $f\in$ \textsc{Ffd}$(d)$.

Diakonikolas et al.~\cite{DiakonikolasLMORSW07}, presented an exponential time testing algorithm for Boolean functions with Fourier degree at most $d$ under the uniform distribution with $\tilde O(2^{6d}/\epsilon^2)$ queries. Later, Chakraborty et al.~\cite{ChakrabortyGM11} improved the query complexity to $\tilde O(2^{2d}/\epsilon^2)$. Bshouty presented a $poly(2^d,n)$ time testing algorithm with $\tilde O(2^{2d}+2^d/\epsilon)$ queries. 

Here we prove\footnote{All the $\tilde O(\cdot)$ in Lemma~\ref{LemFFd} and Lemma~\ref{SPoDd} are independent of $\epsilon$.} 
\begin{lemma}\label{LemFFd}
There is a $poly(2^d,n)$-time testing algorithm for functions with Fourier degree at most $d$ with $$\tilde O(2^{2d})+O\left(\frac{1}{\epsilon}\right)$$ queries.
\end{lemma}
\begin{proof}
Bshouty presented in~\cite{Bshouty18} an exact learning algorithm $A$ for such a class\footnote{The class in~\cite{Bshouty18} refers to the class of decision trees of depth $d$, but the analysis also applies to the class of functions with Fourier degree at most $d$.}. This algorithm makes $M=\tilde O(2^{2d}\log n)$ queries for any constant confidence parameter $\delta$. In Lemma~\ref{Memb}, we show that $($\textsc{Ffd}$(d),A)$ is membership testable in time $O(2^{2d})$. See Definition~\ref{Defini}.

By Bshouty algorithm and Lemma~\ref{LearningToTesting}, the class of functions \textsc{Ffd}$(d)[k]$, where $k:=O(2^d)$, is testable with $\tilde O(2^{2d})$ queries. 

We now compute $\mu($\textsc{Ffd}$(d))$. Let $f\in$\textsc{Ffd}$(d)$ be any function that depends on $x_j$. Then
$$
    \Pr\big[f_{|x_j \gets -1} \neq f_{|x_j \gets 1} \big] =$$
     $$\Pr\left[
        \sum_{S \subseteq [n] \setminus \{j\}}
        \left( \hat{f}_S - \hat{f}_{S \cup \{j\}} \right) \chi_S(x)
        \neq 
        \sum_{S \subseteq [n] \setminus \{j\}}
        \left( \hat{f}_S + \hat{f}_{S \cup \{j\}} \right) \chi_S(x)
    \right]$$
    \begin{eqnarray}
    = \Pr\left[
        \sum_{S \subseteq [n] \setminus \{j\}} 
        \hat{f}_{S \cup \{j\}} \chi_S(x) \neq 0
    \right]. &&\label{jjjjj}
\end{eqnarray}
Let $g=\sum_{S \subseteq [n]\backslash \{j\} } \hat f_{S\cup \{j\}}\chi_S(x)$. Since $f$ depends on $x_j$, we know $\Pr[f_{|x_j\gets -1}\not=f_{|x_j\gets 1}]>0$, which implies there exists some $S\subseteq [n]\backslash\{j\}$ such that $\hat f_{S\cup\{j\}}\not=0$. Let $S_0$ be a set with maximal size for which $\hat f_{S_0\cup\{j\}}\not=0$. Let $V=\{x_i|i\not\in S_0\}$. Assigning arbitrary values in $\{-1,1\}$ to the variables in $\overline{V}$ of $g$ results in a non-zero function. This is due to the uniqueness of the Fourier representation and the fact that no other terms in $g$ can cancel the nonzero term $\hat f_{S_0\cup\{j\}}\chi_S(x)$. Moreover, the resulting function depends on at most $d$ variables because $f$, and hence $g$, has Fourier degree $d$. Thus, the probability that the resulting function is zero for a random uniform assignment to the variables in $\overline{V}$ is at least $1/2^d$. Consequently, the probability in (\ref{jjjjj}) is at least $1/2^d$, implying that $\mu($\textsc{Ffd}$(d))\ge 1/2^{d}$.

 Now, by Lemma~\ref{Result2}, there exists a poly$(2^d,n)$ time algorithm with 
$$\tilde O(2^{2d})+O\left(\frac{d2^d}{1/2^{d}}+\frac{d^32^d}{\log d}+\frac{1}{\epsilon}\right)=\tilde O(2^{2d})+O\left(\frac{1}{\epsilon}\right)$$
queries.
\end{proof}

Bshouty presented in~\cite{Bshouty18} an exact learning algorithm $A$ for such a class. This algorithm makes $M=\tilde O(2^{2d}\log n)$ queries for any constant confidence parameter $\delta$. The algorithm finds the nonzero Fourier coefficients $\hat f_S$ for all $|S|\le d$ and outputs $\sum_{|S|\le d}\hat f_S\chi_S(x)$. We now prove
\begin{lemma}\label{Memb}
    $($\textsc{Ffd}$(d),A)$ is membership testable in time $O(2^{2d})$.
\end{lemma}
\begin{proof}
By the definition of membership testable (Definition~\ref{Defini}), we need to show the following: Given $g=\sum_{|S|\le d}\hat f_S\chi_S(x)$, decide whether $g$ is a Boolean function.

First, since Boolean functions of degree at 
most $d$ have at most $2^{2d-d}$ non-zero Fourier coefficients, if $g$ has more than $2^{2d-2}$ Fourier coefficients, then $g$ is not a Boolean function. If $g$ has fewer than $2^{2d-2}$ coefficients, then $g$ is a Boolean function if and only if $g^2=1$. That is, $g(0^n)=1$ and for every $C\not=\emptyset$, $\sum \hat f_A\hat f_{C\Delta A}=0$. Since there are at most $O(2^{2d})$ coefficients, this can be performed in time $O(2^{4d})$. This gives a deterministic algorithm that runs in time $O(2^{4d})$. We now present a randomized algorithm that runs in time $O(2^{2d})$.

It is sufficient to show that if $g(x)$ is not Boolean function, then $\Pr_x[g(x)\not\in\{+1,-1\}]\ge 1/2^{2d}$. Consider $h(x)=g^2(x)-1$. Since each $\chi_{S_1}(x)\chi_{S_2}(x)$ in $g^2$ depends on at most $2d$ variables, it can be expressed as a multivariate polynomial of degree at most $2d$. Therefore, $h(x)$ is also a multivariate polynomial of degree at most $2d$. 

Suppose $h(x)$ is not identically zero over the domain $\{-1,+1\}^n$. Consider a monomial $M(x)$ in $h(x)$ with a maximal number of variables. Then $m:=\deg(M)\le 2d$. 
Substituting any $\{-1,+1\}$ values in the variables that do not occur in $M$ gives a non-zero multivariate polynomial $h(x)$ of degree $m$ that contains $M$. Since $h(x)$ is a non-zero polynomial, the probability that it is not zero is at least $2^{-m}\ge 2^{-2d}$. This completes the proof. 
\end{proof}

\subsection{Testing $s$-Sparse Polynomial of Degree $d$}
A polynomial (over the field $F_2$) is a sum (in the binary field $F_2$) of monotone terms. An $s$-sparse polynomial is a sum of at most $s$ monotone terms. A polynomial $f$ is said to be of degree $d$ if all its terms are monotone $d$-terms\footnote{A monotone $d$-term is a term with at most $d$ variables}. The class $s$-Sparse Polynomial of Degree $d$ consists of all such $s$-sparse polynomials. 

In the uniform distribution model, Diakonikolas et al.~\cite{DiakonikolasLMORSW07}, presented the first testing algorithm for the class $s$-Sparse Polynomial, which runs in exponential time and makes $\tilde O(s^4/\epsilon^2)$ queries. Chakraborty et al.~\cite{ChakrabortyGM11} improved the query complexity to $\tilde O(s/\epsilon^2)$.  Later, Diakonikolas et al.~\cite{DiakonikolasLMSW11} presented the first polynomial-time testing algorithm with $poly(s,1/\epsilon)$ queries. In~\cite{AlonKKLR03}, Alon et al. presented a testing algorithm for Polynomial of Degree $d$ with $O(1/\epsilon+d2^{2d})$ queries. They also show the lower bound $\Omega(1/\epsilon+2^d)$. By combining these results, one can construct a polynomial-time testing algorithm for $s$-Sparse Polynomial of Degree $d$ with $poly(s,1/\epsilon)+\tilde O(2^{2d})$ queries. This can be achieved by first running the algorithm of Alon et al.~\cite{AlonKKLR03} and then the algorithm of Diakonikolas et al.~\cite{DiakonikolasLMSW11},  accepting if both algorithms accept. 

Bshouty~\cite{Bshouty20x} presented a testing algorithm for $s$-Sparse Polynomial of Degree $d$ with $\tilde{O}(s / \epsilon + s \cdot 2^d)$ queries. 

In this paper, we improve upon this result by proving the following.
\begin{lemma}\label{SPoDd}
There exists a testing algorithm for $s$-Sparse Polynomial of Degree $d$ with $$\tilde O\left(2^ds\right)+O\left(\frac{1}{\epsilon}\right)$$ queries.
\end{lemma}
\begin{proof} Let $C$ be the class of all $s$-Sparse Polynomial of Degree $d$. It is well known that $\mu(C)=1/2^d$. 
    In~\cite{Bshouty20x} Lemma~41, Bshouty presented a learning algorithm that properly exactly learns $s$-sparse polynomials of degree $d$ with $O(s2^d\log(ns))$ queries. Let $k:=ds$. Then $C\subseteq k$-\textsc{Junta} and $C[k]$ is properly exactly learnable with $\tilde O(s2^d)$ queries. By Lemma~\ref{LearningToTesting}, $C[k]$ is testable with $\tilde O(s2^d)$ queries. Then the result follows from Lemma~\ref{Result2}.
\end{proof}

\subsection{$s$-Sparse Polynomial}
In the literature, the first testing algorithm for the class $s$-Sparse Polynomial runs in exponential time~\cite{DiakonikolasLMORSW07} and makes $\tilde O(s^4/\epsilon^2)$ queries. Chakraborty et al., \cite{ChakrabortyGM11} then presented another exponential time algorithm with $\tilde O(s/\epsilon^2)$ queries. Diakonikolas et al. presented the first polynomial-time testing algorithm with ${poly}(s, 1 / \epsilon)$ queries. Then Bshouty~\cite{Bshouty20x} presented a polynomial-time testing algorithm with $\tilde O(s^2/\epsilon)$ queries, and in~\cite{Bshouty22AO}, a polynomial-time algorithm with $\tilde O(s/\epsilon)$ when $\epsilon<1/s^{3.404}$.

In this paper, we prove the following.
\begin{lemma}
    There is a testing algorithm for $s$-Sparse Polynomial with
    $$\left(\frac{\tilde O(s^2)}{\epsilon}\right)^{\frac{\log \beta}{\beta}+\frac{4.413}{\beta}+\Theta\left(\frac{1}{\beta^2}\right)}+\tilde O(s)+O\left(\frac{1}{\epsilon}\right),$$
    queries, where $\epsilon=1/s^\beta$ and $\beta>1$.

    In particular, when $\epsilon<1/s^{8.422}$, the testing algorithm makes $$O\left(\frac{1}{\epsilon}\right)$$ queries.
\end{lemma}

\begin{proof}
    By Lemma~\ref{sTF01}, we may assume that the target function is a subset of $\tilde O(s^2)$-\textsc{Junta}. 
    In~\cite{Bshouty22AO}, Bshouty proved that for $\epsilon=1/s^\beta$, $\beta>1$, there is a proper learning algorithm for $s$-sparse polynomial with probability of success at least $2/3$ with
\begin{eqnarray}\label{QU}
Q_1=
\left(\frac{s}{\epsilon}\right)^{\gamma(\beta)+o_s(1)}+O\left(\left(\log\frac{1}{\epsilon}\right)\left(\frac{s}{\epsilon}\right)^{\frac{1}{\beta+1}}\log n\right)
\end{eqnarray}
queries and runs in time $O(q_U\cdot n)$,
where
$$\gamma(\beta)=\min_{0\le \eta\le 1} \frac{\eta+1}{\beta+1}+(1+1/\eta)H_2\left(\frac{1}{(1+1/\eta)(\beta+1)}\right)=\frac{\log \beta}{\beta}+\frac{4.413}{\beta}+\Theta\left(\frac{1}{\beta^2}\right).$$ 
The output hypothesis is an $s$-sparse polynomial with monomials of size at most $O(\log(s/\epsilon))$ and therefore in \textsc{$O(s\log(s/\epsilon))$-Junta}. By Lemma~\ref{RCVforP}, $s$-Sparse Polynomial is $O(s\log(s/\epsilon))$-relevant coordinates verifiable in
$$\left(\frac{O(s^2\log(s/\epsilon))}{\epsilon}\right)^{\gamma(\beta)+o_s(1)}
+O\left(\left(\log\frac{1}{\epsilon}\right)\left(\frac{O(s^2\log(s/\epsilon))}{\epsilon}\right)^{\frac{1}{\beta+1}}\log n\right)+$$ $$O(s\log(s/\epsilon))$$ queries.
This simplifies to $$\left(\frac{\tilde O(s^2)}{\epsilon}\right)^{\gamma(\beta)+o_s(1)}
+\tilde O\left(\left(\frac{\tilde O(s^2)}{\epsilon}\right)^{\frac{1}{\beta+1}}\log n\right)+\tilde O(s).$$
By Lemma~\ref{krcvKk}, for $\alpha=1/2$, and since the target is in $\tilde O(s^2)$-\textsc{Junta}, it is $O(s\log(s/\epsilon))$-relevant blocks verifiable in
\begin{eqnarray}\label{kwsd}
   \left(\frac{\tilde O(s^2)}{\epsilon}\right)^{\gamma(\beta)+o_s(1)}
+\left(\frac{\tilde O(s^2)}{\epsilon}\right)^{\frac{1}{\beta+1}}+\tilde O(s)+O\left(\frac{1}{\epsilon}\right) 
\end{eqnarray} queries.
By Lemma~\ref{LearningToTesting} and (\ref{QU}), $s$-Sparse Polynomial$[\tilde O(s^2)]$ is testable with 
\begin{eqnarray}\label{Klwe}
\left(\frac{s}{\epsilon}\right)^{\gamma(\beta)+o_s(1)}+\tilde O\left(\left(\frac{s}{\epsilon}\right)^{\frac{1}{\beta+1}}\right)+O\left(\frac{1}{\epsilon}\right)
\end{eqnarray}
queries. By (\ref{kwsd}), (\ref{Klwe}), and Lemma~\ref{Gresult}, there is a testing algorithm for $s$-Sparse Polynomial with
$$\left(\frac{\tilde O(s^2)}{\epsilon}\right)^{\gamma(\beta)+o_s(1)}
+\left(\frac{\tilde O(s^2)}{\epsilon}\right)^{\frac{1}{\beta+1}}+\tilde O(s)+O\left(\frac{1}{\epsilon}\right)$$ queries,
and since $1/(\beta+1)<\gamma(\beta)$, this is equal to
$$\left(\frac{\tilde O(s^2)}{\epsilon}\right)^{\gamma(\beta)+o_s(1)}
+\tilde O(s)+O\left(\frac{1}{\epsilon}\right).$$
\end{proof}

\subsection{Two General Results}
In this section, we prove Theorem~\ref{Th1}. 

\noindent
{\bf Theorem~\ref{Th1}. }{\it Let $C\subseteq k$-\textsc{Junta} be a class that is closed under variable permutations and zero-one projection.
    If $C$ is exactly properly learnable with $Q(n)$ queries, then there is a property testing algorithm for $C$ with
    $$q:=Q(O(k^2))+O\left(\frac{k\log^2 k}{\log\log k}+\frac{1}{\epsilon}\right)$$
    queries. 
    
    Furthermore, we have $q=O(1/\epsilon)$ for $$\epsilon\le \frac{1}{Q(O(k^2))+\tilde \Theta (k)}.$$
    If $C$ is exactly learnable (not necessarily properly), the above result also holds, but the testing algorithm will run in exponential time with respect to $k$.}

\begin{proof}
    If $C$ is exactly properly learnable with $Q(n)$ queries, then, by Lemma~\ref{ELtoRBV}, $C$ is $(k,\alpha)$-relevant coordinates verifiable with $Q(O(k^2))+O(1/\epsilon+k/\alpha)$ queries. By Lemma~\ref{LearningToTesting}, $C[k]$ is testable with $Q(k)+O(1/\epsilon)$. Therefore, by Lemma~\ref{Gresult}, with
    $$\alpha=\frac{\log\log k+\log\log\frac{1}{\epsilon}}{\log k\log \frac{1}{\epsilon}}$$
    and by (\ref{Cl4}), there exists a testing algorithm for $C$ with
    \begin{eqnarray*}
    && \left(Q(O(k^2))+O\left(\frac{1}{\epsilon}+\frac{k}{\alpha}\right)\right)+\left(Q(k)+O\left(\frac{1}{\epsilon}\right)\right)+\\&&O\left(k\frac{\log \frac{1}{\epsilon}(\log k+\log\log \frac{1}{\epsilon})}{\log\frac{1}{\alpha}}+\frac{1}{\epsilon}\right)\\
    &=& Q(O(k^2))+O\left( k\frac{\log \frac{1}{\epsilon}(\log k+\log\log \frac{1}{\epsilon})}{\log\log k +\log\log\frac{1}{\epsilon}}+\frac{1}{\epsilon}\right)\\
    &=&Q(O(k^2))+O\left(\frac{k\log^2k}{\log\log k}+\frac{1}{\epsilon}\right)   \ \ \ \ \mbox{by (\ref{Cl4})}
    \end{eqnarray*}
    queries. This proves the result.

    If the class $C$ is exactly learnable, then it is also properly exactly learnable in exponential time. This is because we can learn $h$, which is equivalent to $f$, find the relevant variables as done in the proof of Lemma~\ref{ELtoRCV}, and then exhaustively search for a function in $C[k]$ that is equivalent to $h$. This implies the result.
\end{proof}

Recall the definition $$\mu(C):=\min_{f\in C}\min_{i\in \RC(f)}\Pr_x[f_{|x_i\gets 0}(x)\not=f_{|x_i\gets 1}(x)]$$
where $\RC(f)$ is the set of relevant coordinates in $f$.

We prove.

\noindent
{\bf Theorem \ref{Th2}.}
{\it Let $C\subseteq k$-\textsc{Junta} be a class that is closed under variable permutations and zero-one projections.
    If $C$ is exactly properly learnable with $Q(n)$ queries, then there is a property testing algorithm for $C$ with
    $$q:=Q(k)+O\left(\frac{k}{\mu(C)}+\frac{k\log^2 k}{\log\log k}+\frac{1}{\epsilon}\right)$$
    queries. We have $q=O(1/\epsilon)$ for $$\epsilon\le \frac{1}{Q(k)+\tilde \Theta (k/\mu(C))}.$$

    If $C$ is exactly learnable, the above result holds with exponential time in $k$.
}
\begin{proof}
    Since $C$ is exactly properly learnable with $Q(n)$ queries,  $C[k]$ is properly exactly learnable with $Q(k)$ queries. By Lemma~\ref{LearningToTesting}, $C[k]$ is testable with $Q(k)+1/\epsilon$ queries. Then, the result follows directly from Lemma~\ref{Result2}.
\end{proof}

\end{document}